\documentclass{article}
\usepackage{common}

\title{Rejection Sampling with Vertical Weighted Strips}
\author{Andrew~M. Raim$^*$, James~A. Livsey, and Kyle~M. Irimata
\vspace{0.5em} \\
Center for Statistical Research and Methodology, U.S. Census Bureau
}
\date{}

\begin{document}

\maketitle

\begin{abstract}
A number of distributions that arise in statistical applications can be expressed in the form of a weighted density: the product of a base density and a nonnegative weight function. Generating variates from such a distribution may be nontrivial and can involve an intractable normalizing constant. Rejection sampling may be used to generate exact draws, but requires formulation of a suitable proposal distribution. To be practically useful, the proposal must both be convenient to sample from and not reject candidate draws too frequently. A well-known approach to design a proposal involves decomposing the target density into a finite mixture, whose components may correspond to a partition of the support. This work considers such a construction that focuses on majorization of the weight function. This approach may be applicable when assumptions for adaptive rejection sampling and related algorithms are not met. An upper bound for the rejection probability based on this construction can be expressed to evaluate the efficiency of the proposal before sampling. A method to partition the support is considered where regions are bifurcated based on their contribution to the bound. Several applications based on the von Mises Fisher distribution are presented to illustrate the framework.
\end{abstract}


\blfootnote{%
Disclaimer: This article is released to inform interested parties of ongoing research and to encourage discussion of work in progress. Any views expressed are those of the authors and not those of the U.S.~Census Bureau.
\begin{flushleft}
$^*$For correspondence: \\
Andrew~M. Raim (\url{andrew.raim@census.gov}) \\
Center for Statistical Research and Methodology \\
U.S.~Census Bureau \\
Washington, DC, 20233, U.S.A.
\end{flushleft}
}

\section{Introduction}
\label{sec:intro}

Consider a weighted distribution \citep{PatilRao1978} with density
\begin{align}
f(x) = f_0(x) / \psi, \quad
f_0(x) = w(x) g(x), \quad
\psi = \int_\Omega f_0(x) d\nu(x),
\label{eqn:weighted-target}
\end{align}
whose support is $\Omega$ and $\nu$ is an appropriate dominating measure. The function $g$ is assumed to be a normalized density for the base distribution. The weight function $w(x) \geq 0$ reweights density $g$ on support $\Omega$ in a systematic way. The normalizing constant $\psi$ is expressed relative to \eqref{eqn:weighted-target} such that $g$ itself is assumed to be normalized; $\psi$ may be intractable or impractical to compute. Distributions of the form \eqref{eqn:weighted-target} often arise as targets for which a random sample is desired. For example, in Bayesian analysis, such an $f$ frequently involves a posterior distribution of interest or one of its conditionals, with $g$ arising from a prior distribution on unknown parameter $\vec{\theta}$ and $w$ from a likelihood which depends on $\vec{\theta}$. The classical method of rejection sampling continues to be relevant when an exact draw is desired from the target, rather than a Markov chain whose invariant distribution is the target, and no other method to directly generate draws is apparent. This work revisits the method of vertical strips to construct proposal distributions for rejection sampling with weighted targets of the form \eqref{eqn:weighted-target}. The resulting method provides additional flexibility which may be useful in obtaining useful samplers with relatively little effort or very efficient samplers---in terms of computational burden and probability of rejection---with additional insight into the components.

Rejection sampling \citep{vonNeumann1951} generates variates from a target distribution by utilizing an envelope function which bounds the unnormalized target density from above. This approach samples from the area beneath the envelope, rejects draws which fall above the target density, and produces accepted draws which are independent and identically distributed from the target density \citep{RobertCasella2004, MartinoEtAl2018}. Many types of proposal densities have been introduced to form the envelope. One such method utilizes a stepwise proposal \citep{Ahrens1993, Ahrens1995, PangEtAl2002} which can be regarded as a construction by vertical strips \citep[Section~3.6.1]{MartinoLuengoMiguez2018}. \citet[Chapter VIII]{Devroye1986} discusses what is essentially the vertical strips method although non-adaptive and for log-concave densities. Another class of methods which is based on construction of a proposal by horizontal strips is referred to as the ziggurat method \citep{MarsagliaTsang2000}. This approach uses a set of rectangles to form the envelope, with the accuracy of the approximations improving as the number of rectangles increases \citep[Section~3.6.1]{MartinoLuengoMiguez2018}.

In practice, selection of an appropriate envelope for rejection sampling faces two main challenges: the envelope must be guaranteed to be an upper bound for the target density, and selection of too large an envelope will yield an inefficient sampler with many of the proposed draws rejected. Adaptive rejection sampling (ARS) addresses these challenges for log-concave targets by automatically adapting to the target using rejected draws, thus yielding an envelope that provides an increasingly tight bound \citep{GilksWild1992}. ARS and related methods typically construct envelopes based on the log-density of the target, composing linear components or other simple functions to ensure that rejection sampling is tractable \citep[Section~4.3]{MartinoLuengoMiguez2018}. The Adaptive Rejection Metropolis Sampling (ARMS) drops the log-concave restriction and uses a Metropolis step \citep{GilksBestTan1995}; however, it produces a chain of non-independent draws and the proposal is not guaranteed to converge to the target with adaptations. The adaptive independent sticky MCMC \citep{MartinoEtAl2018} extends ARMS using non-parametric adaptive proposal densities to reduce the computational burden and improve convergence. The Independent Doubly Adaptive Rejection Metropolis Sampling (IA2RMS) algorithm addresses the ARMS convergence issue and reduces dependence \citep{MartinoReadLuengo2015}. \citet{EvansSwartz1998} propose a sampler which relaxes the log-concavity requirement, and requires that a given transformation of the target density is concave. Another variation of this idea is the convex-concave ARS introduced by \citet{GorurTeh2011} which decomposes the target distribution into concave and convex functions. In addition to the above, many other adaptive rejection samplers have been introduced in the literature; \citet{MartinoLuengoMiguez2018} provide a summary of some of the most common methods.

One of the main challenges in adaptive rejection sampling lies in the construction of the sequence of proposal densities. 
The construction must satisfy three requirements: 
(1) must provide an upper bound for the target density for all $x$ in the domain, 
(2) must be possible to sample exactly from, and 
(3) must converge towards the target density as the number of support points goes to infinity. 
Satisfying these three criteria can be challenging, especially in the multivariate case \citep[Section~4.2]{MartinoLuengoMiguez2018}. 

In the case of weighted distributions of the form \eqref{eqn:weighted-target}, direct sampling as originally proposed by \citet{WalkerEtAl2011} offers an appealing alternative to rejection sampling. This approach reformulates the target by introducing an auxiliary variable; sampling sequentially from the marginal density of the auxiliary variable---which is monotonically nonincreasing with support on the unit interval---and the conditional density given the auxiliary variable may be more tractable than sampling from the original target distribution. \citet{DirectSamplingStep2023} utilizes a step function with the direct sampling approach to approximate the distribution of the auxiliary variable to a desired tolerance. The step function may also be used as an envelope for rejection sampling to obtain an exact sample. One challenge that arises with the direct sampling approach is that the distribution of the latent variable may be focused on an extremely narrow interval; therefore, computations must be implemented carefully to avoid loss of numerical precision. The methods in the present paper can be used in many of the same scenarios as the direct sampler from \citet{DirectSamplingStep2023} and are more straightforward to implement; moreover, it is demonstrated that direct sampling with a step function can be considered an special case of this approach.

This paper introduces the vertical weighted strips (VWS) method of constructing proposals for weighted densities. VWS is an extension of vertical strips that utilizes decomposition \eqref{eqn:weighted-target} of the target density as a weighted form; this provides flexibility to construct efficient proposal distributions which are also convenient for use in a rejection sampler. The method is based on finding an appropriate majorizer for the weight function; i.e., a function which serves as an upper bound. We present two specific variations: one utilizing a constant function on each subset in a partition of the support and the other using linearity to bound a weight function which is either log-convex or log-concave on each subset. Both approaches are demonstrated to yield practical samplers in several illustrations, with the linear majorizer achieving higher efficiency but requiring a conjugacy between the majorizer and the base distribution to be practical. Note that the VWS approach does not require that the target density itself is log-concave as in the original ARS algorithm. 

The rest of this paper is organized as follows. Section~\ref{sec:vertical-weighted-strips} reviews rejection sampling and introduces VWS as a method to construct proposals. Section~\ref{sec:design} discusses considerations in the design of a VWS proposal. Section~\ref{sec:vmf} highlights uses of VWS in several examples involving the von Mises Fisher (VMF) distribution. Section~\ref{sec:conclusions} gives concluding remarks. Appendices supply proofs of propositions and auxiliary material.

\section{Vertical Weighted Strips}
\label{sec:vertical-weighted-strips}

To generate draws from target $f$ in \eqref{eqn:weighted-target}, rejection sampling requires a proposal density $h(x) = h_0(x) / \psi_*$ whose normalizing constant is $\psi_* = \int_\Omega h_0(x) d\nu(x)$, and $M > 0$ is a ratio adjustment factor such that
\begin{align}
\sup_{x \in \Omega} f_0(x) / h_0(x) \leq M.
\label{eqn:ratio-adjustment}
\end{align}
A proposal consisting of variates $u$ and $x$ is generated from $\text{Uniform}(0,1)$ and $h$, respectively. The proposed $x$ is accepted as a draw from $f$ if $u \leq f_0(x) / \{ M \cdot h_0(x) \}$; otherwise, it is rejected and the process is repeated by redrawing $u$ and $x$. This procedure may be repeated $n$ times to obtain an independent and identically distributed sample $x_1, \ldots, x_n$ from $f$. A desirable choice of $h$ is one whose support contains $\Omega$, where $h_0$ is easy to evaluate, which is easy to draw variates from, and whose density is distributed in a manner not too different than $f$. With this formulation, it is routine to show that the probability of accepting a proposed $x$ with accompanying $u$ is
\begin{align}
\Prob\left( U \leq \frac{f_0(X)}{M h_0(X)} \right)
&= \frac{\psi}{M \psi_*},
\label{eqn:rejection-probability}
\end{align}
where $X \sim h$ and $U \sim \text{Uniform}(0,1)$, and that the distribution of an accepted draw is indeed the target distribution; i.e.,
\begin{math}
\Prob(X \in A \mid U \leq f_0(X) / \{M h_0(X)\})
= \int_A f(x) d\nu(x),
\end{math}
for any measureable set $A$ in $\Omega$. Let $S_i$ be the number of draws needed to accept the $i$th variate for $i = 1, \ldots, n$; $\sum_{i=1}^n S_i$ is a negative binomial random variable with probability of success $\psi / \{ M \psi_*\}$ and expected value $n M \psi_* / \psi$. It is apparent that the efficiency of a rejection sampler depends on the ratio of normalizing constants $\psi / \psi_*$ and the adjustment factor $M$. Improvements to efficiency may be possible when $h$ is parameterized, say, by $\vec{\vartheta}$, so that a small value of $M$ can be sought, i.e.,
\begin{math}
M = \inf_{\vec{\vartheta}} \Big\{ \sup_{x \in \Omega} f_0(x) / h_0(x \mid \vec{\vartheta}) \Big\}.
\end{math}

Taking weighted distribution $f(x) \propto w(x) g(x)$ as the target, \eqref{eqn:ratio-adjustment} suggests a particular class of proposals of the form
\begin{math}
h_0(x) = \overline{w}(x) g(x)
\end{math}
for some function $\overline{w}$ which majorizes the weight function $w$ on $\Omega$. That is, $\overline{w}(x) \geq w(x)$ for all $x \in \Omega$. By this construction, $f_0(x) \leq h_0(x)$ for all $x \in \Omega$ so that the adjustment ratio $M$ may be taken to be 1. We may anticipate that the rejection rate $1 - \psi / \psi_*$ will be lower when $\overline{w}$ is closer to $w$. However, we must also be able to readily generate variates from the resulting distribution $h$ for it to be useful as a proposal. It may also be desirable to sequentially refine $\overline{w}$ to be closer to $w$, perhaps at the cost of additional computation and/or labor by the practitioner.

In particular, consider partitioning $\Omega$ into $N$ disjoint regions $\mathscr{D}_1, \ldots, \mathscr{D}_N$, and suppose there are corresponding functions $\overline{w}_j$ such that
\begin{math}
\overline{w}_j(x) \geq w(x)$ for all $x \in \mathscr{D}_j
\end{math}
for each region $j = 1, \ldots, N$. The choice of majorizer
\begin{math}
\overline{w}(x) = \sum_{j=1}^N \overline{w}_j(x) \ind(x \in \mathscr{D}_j)
\end{math}
yields
\begin{math}
h_0(x) = g(x) \sum_{j=1}^N \overline{w}_j(x) \ind(x \in \mathscr{D}_j).
\end{math}
Define
\begin{math}
\overline{\xi}_j = \int_{\mathscr{D}_j} \overline{w}_j(x) g(x) d\nu(x)
\end{math}
and let $\psi_N = \sum_{j=1}^N \overline{\xi}_j$. The normalized proposal $h$ is a finite mixture 
\begin{align*}
h(x) &= h_0(x) / \psi_N
%
%
= \sum_{j=1}^N \pi_j g_j(x),
\end{align*}
whose component densities
\begin{math}
g_j(x) = 
\overline{w}_j(x) g(x) \ind(x \in \mathscr{D}_j) / \overline{\xi}_j
\end{math}
are truncated and reweighted variants of base distribution $g$ and whose mixing weights are $\pi_j = \overline{\xi}_j / \{ \sum_{\ell=1}^N \overline{\xi}_\ell \}$. The dependence on $N$ in the notation $\psi_N$ is emphasized for the upcoming discussion, but it is understood that other terms in the formulation of the proposal depend on $N$ as well. We will refer to the rejection sampling method with proposal $h$ as vertical weighted strips. Generating a variate from $h$ can be accomplished using its finite mixture formulation by drawing index $j$ from a discrete distribution on values $1, \ldots, N$ with probabilities $\pi_1, \ldots, \pi_N$, then drawing $x$ from the truncated and reweighted base distribution $g_j$.

From \eqref{eqn:rejection-probability}, the probability that a draw from a VWS proposal is rejected is $1 - \psi / \psi_N$. Minimizing this probability is ideal from the perspective of avoiding rejections, but moderate values such as $1/2$ may be satisfactory in many applications. To reduce this probability to a suitable level, a practitioner can take actions such as refining $\mathscr{D}_1, \ldots, \mathscr{D}_N$ into a finer partition, refactoring the weight/base decomposition, or considering different classes of majorizers for $\overline{w}_j$. When the normalizing constant $\psi$ is intractable, an upper bound for $1 - \psi / \psi_N$ can be considered instead. Efforts to formulate the proposal may focus on controlling the bound which will ensure that $1 - \psi / \psi_N$ is also controlled. Suppose $\underline{w}_j$ is a minorizing function of $w$ so that
\begin{math}
0 \leq \underline{w}_j(x) \leq w(x)$ for all $x \in \mathscr{D}_j,
\end{math}
$j = 1, \ldots, N$, and let $\underline{\xi}_j = \int_{\mathscr{D}_j} \underline{w}_j(x) g(x) d\nu(x)$. Propositions~\ref{result:rejection-probability-bound-fmm} and \ref{result:total-variation} are straightforward to prove but are stated as results because of their utility. Proofs are given in Appendix~\ref{sec:proofs}.
\begin{proposition}
\label{result:rejection-probability-bound-fmm}
Under VWS, the probability \eqref{eqn:rejection-probability} of a proposed draw being rejected is bounded above by
\begin{align}
\textstyle 1 - \left\{ \sum_{j=1}^N \underline{\xi}_j \right\} \Big/ \left\{ \sum_{j=1}^N \overline{\xi}_j \right\}.
\label{eqn:bound}
\end{align}
\end{proposition}

\begin{remark}
\label{remark:trivial-minorizer}
Using the trivial minorizer $\underline{w}_j(x) = w(x) \cdot \ind(x \in \mathscr{D}_j)$ yields $\psi \equiv \sum_{j=1}^N \underline{\xi}_j$ so that \eqref{eqn:bound} is equivalent to the actual rejection probability $1 - \psi/\psi_N$.
\end{remark}

The rejection probability $1 - \psi / \psi_N$ may also be interpreted as a relative error in approximating the normalizing constant $\psi$ by the normalizing constant $\psi_N$. If the distribution $h$ can be designed in such a way that this relative error is small, the following result shows that probabilities computed under the proposal will be close to probabilities computed under the target. This suggests that $h$ may be useful directly as an approximation to $f$, aside from rejection sampling.

\begin{proposition}
\label{result:total-variation}
Let $\mathcal{B}$ denote the collection of measurable subsets of $\Omega$, $X \sim f$, and $\tilde{X} \sim h$; then
\begin{align}
\sup_{B \in \mathcal{B}} \left| \Prob(\tilde{X} \in B) - \Prob(X \in B) \right|
\leq 1 - \frac{\psi}{\psi_N}.
\label{eqn:rectangles}
\end{align}

\end{proposition}

The remainder of the paper will focus on the case where $f$ is a univariate target. Here, $\Omega$ is a subset of the real line, and we will further assume that regions $\mathscr{D}_j$ take the form of intervals $(\alpha_{j-1}, \alpha_j]$ for $j = 1, \ldots, N$, and $\Omega \equiv (\alpha_0, \alpha_N]$.

\section{Proposal Design}
\label{sec:design}

A VWS proposal $h$ is a finite mixture of reweighted and truncated variants of the base density $g$ on the partition $\mathscr{D}_1, \ldots, \mathscr{D}_N$; several elements of its design should be considered for it to be useful in practice. One must decide on a decomposition of the target $f$ into $w$ and $g$, select a majorizer $\overline{w}$ and minorizer $\underline{w}$, and determine a method to refine the proposal if $N=1$ is not sufficient.

It is desirable that $h$ generates draws with a low rejection probability. This can be achieved using a majorizer $\overline{w}$ which is a close upper bound to the original $w$. However, it is crucial that the proposal be practical to formulate and compute, with mixture components $g_j$ easy to draw from and constants $\overline{\xi}_j$ and $\underline{\xi}_j$ easy to compute. When refining the proposal, it is desirable to avoid creating many regions of low relevance: those where $\overline{\xi}_j$ is small relative to the sum $\sum_{\ell=1}^N \overline{\xi}_\ell$ will rarely be utilized in sampling. Computational effort required to determine the partitioning must be factored into the overall workload. A less optimal partitioning may be preferred if the resulting acceptance rate is satisfactory and rejections are cheap, as opposed to a more optimal partitioning that requires significantly more time and computation to prepare.

The decomposition of $f$ into $w$ and $g$ in \eqref{eqn:weighted-target} is not unique. A natural starting point is to identify the functional form of $g$ by inspecting $f$ and taking $w$ as the remaining factor. However, for any function $q$ which is positive on $\Omega$,
\begin{align}
f_0(x)
= w(x) \frac{1}{q(x)} \cdot q(x) g(x)
= \tilde{w}(x) \tilde{g}(x),
\label{eqn:refactor}
\end{align}
so that another valid base density is proportional to $\tilde{g}(x) = q(x) g(x)$ with $\tilde{w}(x) = w(x) / q(x)$ taken as the weight function. The role of $q$ is reminiscent of the instrumental distribution in importance sampling \citep[e.g.][Chapter~3]{RobertCasella2004}. We may select $q$ to facilitate majorization of $w$ and obtain a practical form of $g_j$. The choice of a nontrivial $q$ can also be helpful in a situation where the density of $g$ is far removed from that of $f$ in extreme cases where the two distributions are practically on disjoint subsets of $\Omega$; otherwise, numerical issues may arise in this situation due to the very small probabilities involved. Transformations of $f$ utilizing a Jacobian may also be considered to facilitate sampler design.

\subsection{Constant Majorizer}
\label{sec:vws-constant}

The following example demonstrates that the standard vertical strips (VS) method is a special case of VWS.

\begin{example}[Vertical Strips]
\label{example:vertical-strips}
If the support $\Omega$ is bounded and $w$ is finite on $\Omega$, take $w(x) = f_0(x)$ to be the entire unnormalized target and $g$ to be the uniform distribution on $\Omega$. The proposal $h$ is a finite mixture of uniform densities
\begin{math}
g_j(x) \propto \ind(x \in \mathscr{D}_j)
\end{math}
with mixing weights based on
\begin{math}
\overline{\xi}_j = \overline{w}_j \Prob(T \in \mathscr{D}_j),
\end{math}
where $T$ is a uniform random variable on $\mathscr{D}_j$. The choice of minorizer $\underline{w}_j = \min_{x \in \mathscr{D}_j} w(x)$ yields
\begin{math}
\underline{\xi}_j = \underline{w}_j \Prob(T \in \mathscr{D}_j).
\end{math}
\end{example}

Therefore, VS represents a particular decomposition of $f$ into $w$ and $g$. The resulting mixture-of-uniforms proposal is quite practical, especially in the univariate case. Introducing the flexibility to choose other decompositions can facilitate the development of more efficient proposals. Furthermore, when $g$ is not restricted to the uniform distribution, it is possible to remove the assumption that $\Omega$ is bounded. Example~\ref{example:vws-constant} describes a variation of VWS which assumes constants for the majorizer and minorizer as in VS but permits $w$ to be a choice other than $f_0$.

\begin{example}[Constant Majorizer]
\label{example:vws-constant}

When $w(x) < \infty$ for $x \in \mathscr{D}_j$, a choice for the majorizer is the constant $\overline{w}_j = \sup_{x \in \mathscr{D}_j} w(x)$. Here, mixture components of proposal $h$ are
\begin{math}
g_j(x) = g(x) \ind(x \in \mathscr{D}_j) / \Prob(T \in \mathscr{D}_j)
\end{math}
and mixing weights are based on $\overline{\xi}_j = \overline{w}_j \Prob(T \in \mathscr{D}_j)$ with $T \sim g$. Taking the minorizer to be $\underline{w}_j = \inf_{x \in \mathscr{D}_j} w(x)$ yields $\underline{\xi}_j = \underline{w}_j \Prob(T \in \mathscr{D}_j)$.
\end{example}

VWS with constant majorizers and minorizers is amenable to being coded in software; given code to evaluate $w$ and implement several operations for distribution $g$ such as the CDF, quantile function, and the density, many of the remaining operations of the VWS proposal can be automated. Sampling from univariate $h$ can be achieved using the inverse CDF method as described in Appendix~\ref{sec:sampling-univariate-proposal}. Numerical optimization can be used to identify $\overline{w}_j$ and $\underline{w}_j$ on each region; in this work we consider several standard methods discussed in Appendix~\ref{sec:constant-majorizer-details}. It is desirable to obtain closed-form solutions when possible because the process of refining $h$ can involve many such optimizations.

The following example demonstrates that the rejection sampler proposed by \citet{DirectSamplingStep2023}, based on the direct sampling algorithm of \citet{WalkerEtAl2011}, can be seen as a special case of VWS.

\begin{example}[Direct Sampler]
\label{example:direct-sampler}
Suppose in \eqref{eqn:weighted-target} that $w$ is finite with $c := \max_{x \in \Omega} w(x)$ and $Z$ is a random variable with $[Z \mid X = x] \sim \text{Uniform}(0, c / w(x))$. Then the joint density of $(X, Z)$ is the product of conditional
\begin{math}
f(x \mid z) \propto g(x) \ind\{ w(x) > z c \}
\end{math}
and marginal
\begin{math}
p(z) = \frac{c}{\psi} p_0(z),
\end{math}
with $p_0(z) = \int \ind\{ w(x) > z c \} g(x) d \nu(x)$. Therefore, a draw from the target $f$ can be obtained by first drawing $z$ from $p$ then $x$ from $f(x \mid z)$. To sample from density $p$, which is non-increasing on the support $[0,1]$, \citet{DirectSamplingStep2023} proposes a step function which majorizes $p_0$ to serve as an envelope for rejection sampling. This is essentially an application of VS from Example~\ref{example:vertical-strips} to density $p$. Equally spaced knots on $[0,1]$ may not yield an effective majorizer when $p$ is positive only within a small neighborhood of zero; better choices for knots take this into account. Because $p$ is non-increasing, the strategy of \citet{Hormann2002} to construct regions having equal probabilities can be considered.
\end{example}

\subsection{Linear Majorizer}
\label{sec:vws-linear}
A summation $h_0$ of constant functions may not produce the most efficient envelope for the target. Perhaps the next natural step is to instead consider linear functions. Here we describe such a construction which is possible when $\Omega$ can be partitioned into regions $\mathscr{D}_j$ where $w$ is either log-convex or log-concave.

A linear majorizer and minorizer are expressed as 
\begin{math}
\overline{w}_j(x) = \exp\{ \overline{\beta}_{0j} + x \overline{\beta}_{1j} \}
\end{math}
and
\begin{math}
\underline{w}_j(x) = \exp\{ \underline{\beta}_{0j} + x \underline{\beta}_{1j} \},
\end{math}
respectively. Suppose $w(x)$ is finite and log-concave on $\mathscr{D}_j = (\alpha_{j-1}, \alpha_j]$; then for $c \in \mathscr{D}_j$,
\begin{align}
\log w(x) &\leq \log w(c) + (x - c) \nabla(c)
\equiv \overline{\beta}_{j0} + \overline{\beta}_{j1} \cdot x
\label{eqn:linear-majorizer}
\end{align}
where $\overline{\beta}_{j0} = \log w(c) - c \cdot \nabla(c)$, $\overline{\beta}_{j1} = \nabla(c)$, and $\nabla(x) = \frac{d}{dx} \log w(x)$. Therefore, the function $\overline{w}_j(x) = \exp\{ \overline{\beta}_{j0} + \overline{\beta}_{j1} \cdot x \}$ is a majorizer for $w(x)$ on $\mathscr{D}_j$. Note that the constant term $\exp\{ \overline{\beta}_{j0} \}$ cancels from the density $g_j$ upon normalization but is needed in formulating $\overline{\xi}_j$ so that the unnormalized $h_0$ majorizes $f_0$. A possible choice of $c$ is considered in Appendix~\ref{sec:linear-majorizer-details}.

Suppose $\log w(x)$ is finite at the endpoints $\alpha_{j-1}$ and $\alpha_j$ of $\mathscr{D}_j$. To obtain an accompanying minorizer, $x$ may be expressed as a convex combination of the endpoints as $x = (1 - \lambda) \alpha_{j-1} + \lambda \alpha_j$ with $\lambda \in [0,1]$ so that $\lambda = (x - \alpha_{j-1}) / (\alpha_j - \alpha_{j-1})$. Concavity of $\log w(x)$ gives
\begin{align}
\log w(x) &\geq (1-\lambda) \log w(\alpha_{j-1}) + \lambda \log w(\alpha_j) \nonumber \\
&= \log w(\alpha_{j-1}) + \frac{x - \alpha_{j-1}}{\alpha_j - \alpha_{j-1}} [ \log w(\alpha_j) - \log w(\alpha_{j-1})] \nonumber \\
&= \underline{\beta}_{j0} + \underline{\beta}_{j1} \cdot x,
\label{eqn:linear-minorizer}
\end{align}
with
\begin{math}
\underline{\beta}_{j0} = \log w(\alpha_{j-1}) - \alpha_{j-1} \underline{\beta}_{j1}
\end{math}
and
\begin{math}
\underline{\beta}_{j1} = \{\log w(\alpha_j) - \log w(\alpha_{j-1}) \} / \{ \alpha_j - \alpha_{j-1} \},
\end{math}
so that the function $\underline{w}_j(x) = \exp\{ \underline{\beta}_{j0} + \underline{\beta}_{j1} \cdot x \}$ is minorizer for $w(x)$ on $\mathscr{D}_j$. In the case that $w$ is log-convex rather than log-concave, the majorizer and minorizer in \eqref{eqn:linear-majorizer} and \eqref{eqn:linear-minorizer} switch roles.

Examples of a constant and linear majorizer are displayed in Figure~\ref{fig:weight-functions}. The linear majorizer typically achieves much lower rejection probability than the constant majorizer as $N$ is increased but requires more effort to program; operations for the sampler typically need to be coded anew for each new problem. 

\begin{figure}
\centering
\begin{subfigure}{0.40\textwidth}
\includegraphics[width=\textwidth]{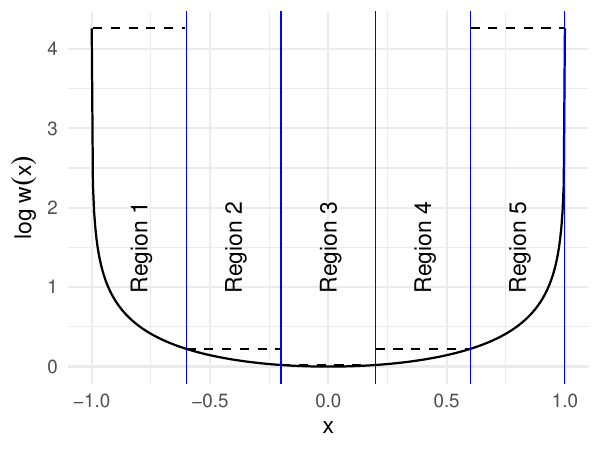}
\caption{}
\label{fig:weight-vws-const}
\end{subfigure}
\begin{subfigure}{0.40\textwidth}
\includegraphics[width=\textwidth]{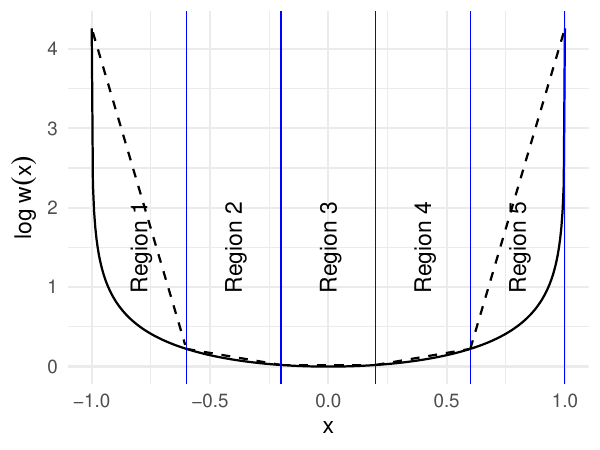}
\caption{}
\label{fig:weight-vws-linear}
\end{subfigure}
\caption{Examples of (\subref{fig:weight-vws-const}) constant and (\subref{fig:weight-vws-linear}) linear majorizers (dashed) of a weight function $w$ (solid).
}
\label{fig:weight-functions}
\end{figure}

In this work, we assume a common form of majorizer and minorizer over the regions $\mathscr{D}_1, \ldots, \mathscr{D}_N$ for a proposal $h$. This is done for convenience and to facilitate implementation of code, but is not a requirement of the methodology. While the choice of decomposition into $w$ and $g$ will typically be fixed within a proposal, majorizer and minorizer forms for $w$ can vary across regions with appropriate bookkeeping in the implementation.

The following examples present cases where a linear majorizer and minorizer yield practical proposals and are useful in Section~\ref{sec:vmf}.

\begin{example}[Exponential Family Base with Linear Majorizer]
\label{example:expfam}
Suppose $w$ is log-convex or log-concave and base distribution $g$ has exponential family density $g(x) = \exp\{ \vartheta x - a(\vartheta) \}$ with respect to dominating measure $\nu$ and $\vartheta \in \mathbb{R}$. Majorizer \eqref{eqn:linear-majorizer} and minorizer \eqref{eqn:linear-minorizer} yield, respectively,
\begin{align*}
&\overline{\xi}_j
= \exp\left( \overline{\beta}_{j0} \right)
\int_{\mathscr{D}_j} \exp\{ (\overline{\beta}_{j1} + \vartheta) x - a(\vartheta) \} d\nu(x), \\
&\underline{\xi}_j
= \exp\left( \underline{\beta}_{j0} \right)
\int_{\mathscr{D}_j} \exp\{ (\underline{\beta}_{j1} + \vartheta) x - a(\vartheta) \} d\nu(x).
\end{align*}
Here, proposal mixture components
\begin{align*}
g_j(x) &= \frac{
\exp\{ (\overline{\beta}_{j1} + \vartheta) x - a(\vartheta) \}
}{
\int_{\mathscr{D}_j} \exp\{ (\overline{\beta}_{j1} + \vartheta) s - a(\vartheta) \} d\nu(s)
} \ind(x \in \mathscr{D}_j)
\end{align*}
are members of the same family as $g$, but truncated to the intervals $(\alpha_{j-1}, \alpha_j]$.
\end{example}

\begin{example}[Doubly-Truncated Exponential Base with Linear Majorizer]
\label{example:doubly-truncated-exp}
Let $X \sim \text{Exp}_{(a,b)}(\kappa)$ denote a random variable with doubly truncated exponential distribution whose density is
\begin{align*}
g(x) = \frac{\kappa e^{\kappa x}}{e^{\kappa b} - e^{\kappa a}} \cdot \ind(a < x < b),
\end{align*}
where $-\infty < a < b < \infty$ and rate $\kappa$ may be any real number. Draws from $\text{Exp}_{(a,b)}(\kappa)$ may be generated with the inverse CDF method, where the CDF and associated quantile function are, respectively,
\begin{align}
&G(x) = \frac{e^{\kappa x} - e^{\kappa a}}{e^{\kappa b} - e^{\kappa a}}, \quad x \in (a,b),
\label{eqn:vmf-base-cdf} \\
&G^{-1}(\varphi) = \frac{1}{\kappa} \log\left[e^{\kappa a} + \varphi (e^{\kappa b} - e^{\kappa a}) \right], \quad \varphi \in [0,1].
\label{eqn:vmf-base-quantile}
\end{align}
Consider using $g$ as a base distribution with majorizer \eqref{eqn:linear-majorizer} and minorizer \eqref{eqn:linear-minorizer}; expressions obtained in Example~\ref{example:expfam} for exponential families become
\begin{align*}
&\overline{\xi}_j
= \frac{\kappa \exp\{ \overline{\beta}_{j0} \}}{(\kappa + \overline{\beta}_{j1}) (e^\kappa - e^{-\kappa})} \left\{
\exp\{(\kappa + \overline{\beta}_{j1}) \alpha_j\} - \exp\{(\kappa + \overline{\beta}_{j1}) \alpha_{j-1} \}
\right\}, \\
&\underline{\xi}_j
= \frac{\kappa \exp\{ \underline{\beta}_{j0} \}}{(\kappa + \underline{\beta}_{j1}) (e^\kappa - e^{-\kappa})} \left\{
\exp\{(\kappa + \underline{\beta}_{j1}) \alpha_j\} - \exp\{(\kappa + \underline{\beta}_{j1}) \alpha_{j-1} \}
\right\}.
\end{align*}
The $j$th component of finite mixture $h$ is
\begin{align*}
g_j(x)
&= \frac{
(\kappa + \overline{\beta}_{j1}) \exp\{(\kappa + \overline{\beta}_{j1}) x\}
}{
\exp\{(\kappa + \overline{\beta}_{j1}) \alpha_j\} - \exp\{(\kappa + \overline{\beta}_{j1}) \alpha_{j-1} \}
}
\cdot \ind(\alpha_{j-1} < x \leq \alpha_j),
\end{align*}
which is the density of
\begin{math}
T \sim \text{Exp}_{(\alpha_{j-1}, \alpha_j]}(\kappa + \overline{\beta}_{j1}).
\end{math}
\end{example}

\begin{example}[Uniform Base with Linear Majorizer]
\label{example:uniform}
Suppose $w$ is either log-convex or log-concave on each $\mathscr{D}_j$ and base distribution $g$ has uniform density $g(x) = \ind(x \in [a,b]) / (b - a)$ so that $\Omega = [a, b]$ is also the support of the target. This corresponds to the VS setting in Example~\ref{example:vertical-strips} for which we may also consider use of a linear majorizer and minorizer. Majorizer \eqref{eqn:linear-majorizer} and minorizer \eqref{eqn:linear-minorizer} give, respectively,
\begin{align*}
&\overline{\xi}_j
= \frac{
\exp\{ \overline{\beta}_{j0} \}
}{
(b - a) \overline{\beta}_{j1}
}
\Big(
\exp\{ \overline{\beta}_{j1} \cdot \alpha_j \} -
\exp\{ \overline{\beta}_{j1} \cdot \alpha_{j-1} \}
\Big), \\
&\underline{\xi}_j
= \frac{
\exp\{ \underline{\beta}_{j0} \}
}{
(b - a) \underline{\beta}_{j1}
}
\Big(
\exp\{ \underline{\beta}_{j1} \cdot \alpha_j \} -
\exp\{ \underline{\beta}_{j1} \cdot \alpha_{j-1} \}
\Big).
\end{align*}
The $j$th component of finite mixture $h$ becomes
\begin{align*}
g_j(x)
= \frac{
\overline{\beta}_{j1} \cdot \exp\{ \overline{\beta}_{j1} \cdot x \}
}{
\exp\{ \overline{\beta}_{j1} \cdot \alpha_j \} -
\exp\{\overline{\beta}_{j1} \cdot \alpha_{j-1} \}
} \ind(\alpha_{j-1} < x \leq \alpha_j),
\end{align*}
which is the density of
\begin{math}
T \sim \text{Exp}_{(\alpha_{j-1}, \alpha_j]}(\overline{\beta}_{j1}).
\end{math}
\end{example}

\subsection{Knot Selection}
\label{sec:knot-selection}

An important consideration in achieving a satisfactory acceptance rate is the method of selecting the knots $\alpha_1, \ldots, \alpha_{N-1}$ which partition domain $\Omega$ into regions $\mathscr{D}_j = (\alpha_{j-1}, \alpha_j]$ for $j = 1, \ldots, N$. It is desirable that the rejection rate $1 - \psi / \psi_N$ reduces rapidly as $N$ increases. If a very large $N$ is required, the effort to prepare the proposal and draw variates may not be worth the efficiency achieved in the final sampler. In this work, we consider a rule of thumb which directly seeks to reduce bound \eqref{eqn:bound}. The contribution of the $\ell$th region to \eqref{eqn:bound} can be characterized as
\begin{math}
\rho_\ell =
\{ \overline{\xi}_\ell - \underline{\xi}_\ell \} \big/ \sum_{j=1}^N \overline{\xi}_j,
\end{math}
so that $\rho_1, \ldots, \rho_{N}$ sum to \eqref{eqn:bound}. We iteratively refine the partition by selecting a region with probability proportional to $\rho_1, \ldots, \rho_{N}$ and bifurcate the selected region. This is described in Algorithm~\ref{alg:refine}. This method tends to select the largest contributors, but allows all regions with $\rho_\ell > 0$ to be selected. We allow the algorithm to complete with less than $N$ regions if the bound \eqref{eqn:bound} is smaller than a given tolerance $\epsilon > 0$.

We have opted to select knots entirely before sampling in this work; however, it is possible to refine the proposal during sampling by identifying the region $\mathscr{D}_{\ell}$ which contains a rejected draw $x$ and splitting that region at $x$. It is also possible to delete knots which yield regions with small contributions $\rho_\ell$ to \eqref{eqn:bound}. Several additional strategies are considered in Appendix~\ref{sec:vmf-knot-selection-study}: equally spaced knots, regions with equal probabilities, and Algorithm~\ref{alg:refine} with a deterministic ``greedy'' selection on line~\ref{line:selection}.

\begin{algorithm}
\caption{Probabilistic rule of thumb for sequential knot selection.}
\label{alg:refine}

\hspace*{\algorithmicindent}\textbf{Input}: maximum number of knots to add $N$. \\
\hspace*{\algorithmicindent}\textbf{Input}: initial vector of internal knots $\alpha_1, \ldots, \alpha_{N_0-1}$; may be empty with $N_0 = 0$. \\
\hspace*{\algorithmicindent}\textbf{Input}: tolerance $\epsilon > 0$.
\begin{algorithmic}[1]
\State $j \leftarrow 0$ 
\While{$j \leq N$}
\State Let $\mathscr{D}_\ell = (\alpha_{\ell-1}, \alpha_{\ell}]$ and compute $\rho_\ell$ for $\ell \in \{ 1, \ldots, N_0 + j \}$.
\State If $\sum_{\ell=1}^{N_0 + j} \rho_\ell < \epsilon$, break from the loop.
\State Draw $\ell \in \{ 1, \ldots N_0 + j \}$ from $\text{Discrete}(\rho_1, \ldots, \rho_{N_0 + j})$. \label{line:selection}
\State Let $\alpha^*$ be the midpoint of $\alpha_{\ell-1}$ and $\alpha_\ell$; add $\alpha^*$ to vector of knots.
\State Let $j \leftarrow j + 1$. \label{line:midpoint}
\EndWhile
\State \Return $(\alpha_0, \ldots, \alpha_{N_0 + j})$.
\end{algorithmic}
\end{algorithm}

\subsection{Restrictions on $w$}
\label{sec:restictions}

We have noted several restrictions on $w$; namely, it must be finite to use the constant or linear majorizer and must be strictly positive to use the linear minorizer (in the log-concave case) or linear majorizer (in the log-convex case). In this work, such problematic cases occur only at the endpoints of $\Omega$. We take a simple approach of truncating the support to a bounded interval $(\alpha_0, \alpha_N]$ within the original $\Omega$ that excludes such endpoints. A more elaborate approach which avoids truncation is discussed in Remark~\ref{remark:infinite-log-weight}.

\section{Illustrations with von Mises Fisher Distribution}
\label{sec:vmf}

The von Mises Fisher (VMF) distribution provides several opportunities to illustrate the vertical weighted strips approach. VMF arises in the study of directional data which are observed on the $d$-dimensional sphere $\mathbb{S}^d = \{ \vec{v} \in \mathbb{R}^d : \vec{v}^\top \vec{v} = 1 \}$. \citet{FisherLewisEmbleton1993} and \citet{MardiaJupp1999} give comprehensive treatments in this area and \citet{PewseyGarciaPortugues2021} provide a survey of more recent developments. A random variable $\vec{V}$ with distribution $\text{VMF}_d(\vec{\mu}, \kappa)$ has density
\begin{align*}
f_{\text{VMF}}(\vec{v}) = \frac{
\kappa^{d / 2 - 1}
}{
(2 \pi)^{-d/2} I_{d/2 - 1}(\kappa)
} \exp(\kappa \cdot \vec{\mu}^\top \vec{v}) \cdot \ind(\vec{v} \in \mathbb{S}^d),
\end{align*}
with modified Bessel function of the first kind
\begin{math}
I_{\nu}(x) = \sum_{m=0}^\infty \{m! \cdot \Gamma(m + \nu + 1)\}^{-1} (\frac{x}{2})^{2m + \nu}.
\end{math}
Parameter $\vec{\mu} \in \mathbb{S}^d$ determines the orientation on the sphere and parameter $\kappa > 0$ determines the concentration. This section will consider the VWS approach in several von Mises Fisher scenarios. Section~\ref{sec:vmf-generation} demonstrates variate generation from VMF. Section~\ref{sec:vmf-probabilities} uses a VWS proposal to approximately compute probabilities---without Monte Carlo---via Proposition~\ref{result:total-variation}. Section~\ref{sec:vmf-bayesian} presents a Bayesian application with independent and identically distributed data observed from a VMF distribution; here, VWS can be utilized to take exact draws from the posterior distribution of the unknown parameters without requiring MCMC. Additional material given in Appendix~\ref{sec:vmf-knot-selection-study} includes a comparison of several knot selection methods.

The terms ``constant VWS'' and ``linear VWS'' will refer to the constant and linear constructions in Example~\ref{example:vws-constant} and Section~\ref{sec:vws-linear}, respectively. The term ``VS'' will refer to the specific factorization in Example~\ref{example:vertical-strips}. Section~\ref{sec:vmf-generation} considers a factorization for VWS which is different than VS, but the two are seen to perform similarly; this is specific to the application and will not necessarily occur in general. We find it instructive to demonstrate both factorizations and see that they yield valid samplers.

The target in Section~\ref{sec:vmf-bayesian} is a case where a non-VS factorization is more immediately practical; here, the target contains a Bessel function which can be incorporated into the weight function and majorized. A remaining exponential term is used as a base density which supports sampling on $(0, \infty)$.

\subsection{Generation of Variates}
\label{sec:vmf-generation}

A widely used method to generate variates from $\text{VMF}_d(\vec{\mu}, \kappa)$ is a rejection sampling scheme developed by \citet{Ulrich1984} and \citet{Wood1994}. For example, this method is used in the R packages movMF \citep{HornikGrun2014} and Rfast \citep{TsagrisPapadakis2018}. The sampler is based on the following construction. Without loss of generality, suppose $\vec{\mu}_0 = (1, 0, \ldots, 0)$. A random variable $\vec{V}_0 \sim \text{VMF}_d(\vec{\mu}_0, \kappa)$ can be obtained using
\begin{math}
\vec{V}_0 = \left( X, \sqrt{1 - X^2} \cdot \vec{U} \right),
\end{math}
where $\vec{U}$ is a uniform random variable on the sphere $\mathbb{S}^{d-1}$ and $X$ has density
\begin{align}
f(x) =
\frac{
(\kappa / 2)^{d/2 - 1} (1 - x^2)^{(d-3)/2} \exp(\kappa x)
}{
\sqrt{\pi} \cdot I_{d/2 - 1}(\kappa) \cdot \Gamma((d-1)/2)
} \cdot \ind(-1 < x < 1).
\label{eqn:vmf-target}
\end{align}
A draw of $\vec{U}$ can be readily obtained from $\vec{Z} / \sqrt{\vec{Z}^\top \vec{Z}}$ with $\vec{Z} \sim \text{N}(\vec{0}, \vec{I}_{d-1})$ \citep{Muller1959}. Furthermore, $\vec{V} \sim \text{VMF}_d(\vec{\mu}, \kappa)$ for an arbitrary $\vec{\mu}$ can be obtained from $\vec{V}_0$ using $\vec{V} = \vec{Q} \vec{V}_0$ with an orthonormal matrix $\vec{Q}$ whose first column is $\vec{\mu}$. Therefore, the problem of drawing $\vec{V}_0$ reduces to univariate generation of $X$.

\citet{Ulrich1984} and \citet{Wood1994} developed a proposal for $X$ based on $Z \sim \text{Beta}((d-1)/2, (d-1)/2)$ via the random variable $X_0 = [1 - (1+b) Z] / [1 - (1-b) Z]$ with density

\begin{align}
h_{X_0}(x \mid b) = \frac{
2 \cdot b^{(d-1)/2} (1 - x^2)^{(d-3)/2}
}{
B((d-1)/2, (d-1)/2) \cdot [(1+b) - (1-b)x]^{d-1}
}, \quad x \in (-1,1),
%
%
\label{eqn:wood-proposal}
\end{align}
where $b \in (0,1)$ is a fixed number. The smallest $M$ such that $f(x) / \{ M h_{X_0}(x \mid b) \} \leq 1$ for all $x \in (-1, 1)$ is obtained from
\begin{align*}
x_* = \frac{1 - b_*}{1 + b_*}, \quad
b_* = \frac{ -2\kappa + \sqrt{4 \kappa^2 + (d-1)^2} }{ d-1 }.
\end{align*}
Let $c = \kappa x_* + (d-1) \log(1 - x_*^2)$. The rejection sampler proceeds by generating $x$ from proposal \eqref{eqn:wood-proposal} and $u$ from $\text{Uniform}(0,1)$. We accept $x$ as a draw from the target if $\log u < \kappa x + (d-1)\log(1 - x \cdot x_*) - c$; otherwise, we reject $x$ and $u$ and draw again. We will refer to this as the UW sampler.

We also note an alternative approach from \citet{KurzHanebeck2015} to sample from target \eqref{eqn:vmf-target}. Here, an inverse CDF method is obtained by deriving expressions for the CDF which are free of integrals and using a bisection algorithm to numerically compute the quantile function.

As a first demonstration of the VWS approach, the following example considers a simple proposal that makes use of inequality $1 - x^2 \leq e^{-x^2}$ without partitioning the support.

\begin{example}[Simple Proposal]
\label{example:simple-proposal}

Consider the decomposition $f_0(x) = w(x) g(x)$ where $w(x) = (1-x^2)^{(d-3)/2}$ and
\begin{math}
g(x) = \kappa e^{\kappa x} \cdot \ind (-1 < x < 1) / (e^\kappa - e^{-\kappa})
\end{math}
is the density of $T \sim \text{Exp}_{(-1,1)}(\kappa)$ from Example~\ref{example:doubly-truncated-exp}. The normalizing constant of $f$ relative to $f_0$ is
\begin{align}
\psi =
\frac{
2^{d/2 - 1} \sqrt{\pi} \cdot I_{d/2 - 1}(\kappa) \cdot \Gamma((d-1)/2)
}{
\kappa^{d/2 - 2} (e^{\kappa} - e^{-\kappa})
}.
\label{eqn:simple-norm-const}
\end{align}
For any $d > 3$, $w(x) = (1-x^2)^{(d-3)/2}$ is majorized by $\overline{w}(x) = e^{-x^2 (d-3) / 2}$. The function $h_0(x) = \overline{w}(x) g(x)$ is recognized as an unnormalized density of a normal random variable with mean $\kappa (d-3)^{-1}$ and variance $(d-3)^{-1}$ which has been truncated to the interval $(-1,1)$. After completing the square and adjusting for the truncated support, the normalized proposal and normalizing constant are, respectively,
\begin{align*}
&h(x) = \sqrt{\frac{d-3}{2\pi}} \exp\left\{-\frac{d-3}{2} \left[x - \kappa (d-3)^{-1} \right]^2 \right\} \cdot
\frac{\ind(-1 < x < 1)}{
\Prob(-1 < T < 1)
}, \\
&\psi_* = \frac{\kappa}{e^{\kappa} - e^{-\kappa}} \sqrt{\frac{2\pi}{d-3}}
\exp\left\{ \frac{1}{2} \kappa^2 (d-3)^{-1} \right\} \cdot
\Prob(-1 < T < 1).
\end{align*}
Drawing from proposal $h$ is straightforward using the inverse CDF method as discussed in Appendix~\ref{sec:sampling-univariate-proposal}. We may therefore proceed with rejection sampling with $h$ as usual. Table~\ref{tab:simple-reweighted} displays the rejection rate $1 - \psi / \psi_*$ for several settings of $d$ and $\kappa$. Acceptance is relatively frequent for smaller $\kappa$ but the sampler becomes increasingly inefficient as $\kappa$ increases beyond 1. It is interesting to note that the rejection rate does not increase monotonically with $d$.

\begin{table}
\centering
\caption{Rejection rates as percentage $100 \times (1 - \psi / \psi_*)$ for the simple VWS sampler.}
\label{tab:simple-reweighted}
\begin{tabular}{rrrrrrrrrr}
\toprule
\multicolumn{1}{c}{} &
\multicolumn{9}{c}{$\kappa$} \\
\cmidrule(lr){2-10}
\multicolumn{1}{c}{$d$} &
\multicolumn{1}{c}{0.1} &
\multicolumn{1}{c}{0.2} &
\multicolumn{1}{c}{0.5} &
\multicolumn{1}{c}{1} &
\multicolumn{1}{c}{2} &
\multicolumn{1}{c}{5} &
\multicolumn{1}{c}{10} &
\multicolumn{1}{c}{20} &
\multicolumn{1}{c}{50} \\
\cmidrule(lr){1-1}
\cmidrule(lr){2-10}
4  &  8.23 &  8.28 &  8.67 &  9.98 & 14.24 & 28.22 & 42.79 & 56.82 & 71.57 \\
5  & 10.76 & 10.83 & 11.32 & 13.01 & 18.73 & 38.95 & 59.70 & 76.62 & 89.76 \\
10 &  8.60 &  8.65 &  8.97 & 10.11 & 14.50 & 38.44 & 73.71 & 94.50 & 99.64 \\
20 &  4.16 &  4.17 &  4.26 &  4.58 &  5.86 & 15.43 & 48.50 & 93.45 & 99.98 \\
50 &  1.56 &  1.56 &  1.58 &  1.62 &  1.82 &  3.23 &  9.33 & 41.17 & 99.86 \\
\bottomrule
\end{tabular}
\end{table}

\end{example}

The following examples use the constant and linear constructions from Sections~\ref{sec:vws-constant} and \ref{sec:vws-linear} which can yield proposals with small rejection rates for a wide range of $\kappa$ and $d$, including cases $d \in \{ 2, 3 \}$ precluded from Example~\ref{example:simple-proposal}. We consider two possible factorizations using a constant majorizer.

\begin{example}[Constant VWS]
\label{example:vmf-constant-vws}
Let $\alpha_0 = -1$ and $\alpha_N = 1$ and decompose the target as
\begin{math}
f_0(x) = w(x) g(x)
\end{math}
with $w(x) = (1 - x^2)^{(d-3)/2}$ and $g$ the density of $\text{Exp}_{(\alpha_0, \alpha_N)}(\kappa)$ from Example~\ref{example:doubly-truncated-exp}. The derivative
\begin{math}
\frac{d}{dx} \log w(x) = -(d-3) \frac{x}{1 - x^2}
\end{math}
is positive for $x \in (-1, 0)$, negative for $x \in (0, 1)$, and has root $x = 0$. Therefore, $\log w(x)$ is unimodal on $(-1, 1)$ with a maximum at $x = 0$. A majorizer and minorizer, respectively, of $w$ on region $\mathscr{D}_j = (\alpha_{j-1}, \alpha_j]$ is given by
\begin{align}
\overline{w}_j = 
\begin{cases}
w(0), & \text{if $\alpha_{j-1} < 0 \leq \alpha_j$}, \\
w(\alpha_{j-1}), & \text{if $\alpha_{j-1} \geq 0$}, \\
w(\alpha_j), & \text{if $\alpha_j < 0$},
\end{cases}
\label{eqn:vmf-vws-constant-majorizer}
\end{align}
and
\begin{align}
\underline{w}_j = 
\begin{cases}
\min\{ w(\alpha_{j-1}), w(\alpha_j) \}, & \text{if $\alpha_{j-1} < 0 \leq \alpha_j$}, \\
w(\alpha_j), & \text{if $\alpha_{j-1} \geq 0$}, \\
w(\alpha_{j-1}), & \text{if $\alpha_j < 0$},
\end{cases}
\label{eqn:vmf-vws-constant-minorizer}
\end{align}
The constants used in bound \eqref{eqn:bound} are
\begin{align}
\overline{\xi}_j = \overline{w}_j \Prob(\alpha_{j-1} < T \leq \alpha_{j})
\quad \text{and} \quad
\underline{\xi}_j = \underline{w}_j \Prob(\alpha_{j-1} < T \leq \alpha_{j}),
\label{eqn:vmf-vws-constant-xi}
\end{align}
where $T \sim \text{Exp}_{(\alpha_{j-1}, \alpha_j)}(\kappa)$.
\end{example}

\begin{example}[Constant VS]
\label{example:vmf-constant-vs}
Let us modify Example~\ref{example:vmf-constant-vws} and decompose the target as
\begin{math}
f_0(x) = w(x) g(x)
\end{math}
with $w(x) = (1 - x^2)^{(d-3)/2} e^{\kappa x}$ and $g$ the density of $\text{Uniform}(\alpha_0, \alpha_N)$. The derivative
\begin{math}
\frac{d}{dx} \log w(x) = -(d-3) \frac{x}{1 - x^2} + \kappa,
\end{math}
differs from Example~\ref{example:vmf-constant-vws} only by a constant, so that forms \eqref{eqn:vmf-vws-constant-majorizer} for the majorizer and \eqref{eqn:vmf-vws-constant-minorizer} for the minorizer also apply here. Constants \eqref{eqn:vmf-vws-constant-xi} for bound \eqref{eqn:bound} are obtained with $T \sim \text{Uniform}(\alpha_{j-1}, \alpha_j)$.
\end{example}

\begin{example}[Linear VWS]
\label{example:vmf-linear-vws}
Let $\alpha_0 = -1 + \epsilon$ and $\alpha_N = 1 - \epsilon$ for a small $\epsilon > 0$, taken to be $10^{-4}$ in the present section and $10^{-6}$ in Section~\ref{sec:vmf-probabilities}. Again consider the decomposition of $f$ in Example~\ref{example:vmf-constant-vws}. For a linear majorizer, note that 
\begin{math}
\frac{d^2}{dx^2} \log w(x) = -(d-3) \frac{1 + x^2}{(1 - x^2)^2};
\end{math}
therefore, $w$ is log-convex if $d < 3$, log-concave if $d > 3$, and a constant if $d = 3$. Accordingly, the majorizer for $w$ on $\mathscr{D}_j$ may be taken as
\begin{math}
\overline{w}_j = \exp\{ \overline{\beta}_{0j} + \overline{\beta}_{1j} x \}
\end{math}
with constants $\overline{\beta}_{0j}$ and $\overline{\beta}_{1j}$ selected according to Section~\ref{sec:vws-linear}. From Example~\ref{example:doubly-truncated-exp}, density $g_j$ corresponds to distribution $\text{Exp}_{(\alpha_{j-1}, \alpha_j)}(\kappa + \overline{\beta}_{j1})$. The expression for $\overline{\xi}_j$ is given in Example~\ref{example:doubly-truncated-exp}. Rather than the corresponding linear minorizer, we use the trivial minorizer discussed in Remark~\ref{remark:trivial-minorizer} here to operate on the exact rejection rate, via
\begin{displaymath}
\underline{\xi}_j = \int_{\alpha_{j-1}}^{\alpha_j} w(x) g(x) dx
= \int_{\alpha_{j-1}}^{\alpha_j} 
\frac{(1 - x^2)^{(d-3)/2} \kappa e^{\kappa x}}{e^{\kappa \alpha_N} - e^{\kappa \alpha_0}} dx
\end{displaymath}
which may be computed numerically.
\end{example}

To generate draws from $f$, let us consider the UW rejection sampler and several variations of the VWS sampler. A small study has been carried out to compare overall rejection rates for the UW, constant VS, constant VWS, and linear VWS methods. We consider $d \in \{ 2, 4, 5 \}$ and $\kappa \in \{ 0.1, 10 \}$. Note that we have skipped $d = 3$ because it is an easier case with the quadratic term vanishing from $f$. In each setting, the UW rejection rate is computed empirically using 50,000 draws. The VS/VWS samplers are based on $N$ regions, where $N$ is refined from a single region to 100 regions using Algorithm~\ref{alg:refine}. The process of refining from one to 100 regions is repeated 100 times to capture randomness used in the selection. Figure~\ref{fig:vmf-adapt} displays the rejection rate on the log-scale as the median value of $\log(1 - \psi / \psi_N)$, along with a confidence band highlighting the 2.5\% and 97.5\% quantiles taken over the 100 repetitions. The UW sampler is seen to be quite efficient when $\kappa = 0.1$ but rejects more frequently when $\kappa = 10$. In the case that $\kappa = 10$ and $d = 5$, the rejection rate is $23.9\%$. The constant VS and constant VWS samplers perform comparably, with constant VWS slightly more efficient for $\kappa = 10$. For all nine settings of $d$ and $\kappa$, constant VS/VWS achieve a rejection rate of $\exp(-2.47) = 8.5\%$ or smaller with 100 regions; this is competitive with the efficiency of UW when $\kappa \geq 1$. Linear VWS outperforms constant VS/VWS in all settings and achieves rejection rates several orders of magnitude smaller. To improve upon the rejection rate of UW, linear VWS requires nearly $N = 100$ when $\kappa = 0.1$ but only a small $N$ when $\kappa = 10$. In this application, the additional effort to derive and implement a linear majorizer yields a proposal which can achieve a very low rejection rate over the family of target distributions.

Constant VWS and constant VS were seen to perform comparably in this example. Furthermore, Appendix~\ref{sec:vmf-generation-extra} shows that the corresponding factorizations yield equivalent proposals using the linear majorizer. Therefore, only the linear construction in Example~\ref{example:vmf-linear-vws} was considered in the study. Under these two factorizations, the two weight functions $\log w(x) = \frac{1}{2} (d-3) \log (1 - x^2)$ and $\log w(x) = \frac{1}{2} (d-3) \log (1 - x^2) + \kappa x$ are equally amenable to majorization \& minorization, as they differ only by a linear term. Moreover, truncating, reweighting, and drawing from the corresponding base distributions, $\text{Exp}_{(\alpha_0, \alpha_N]}(\kappa)$ and $\text{Uniform}(\alpha_0, \alpha_N)$, can be carried out in either case without too much difficulty. The choice of factorization becomes more consequential when it makes such operations substantially more or less practical.

\begin{figure}
\centering
\begin{subfigure}{0.32\textwidth}
\includegraphics[width=\textwidth]{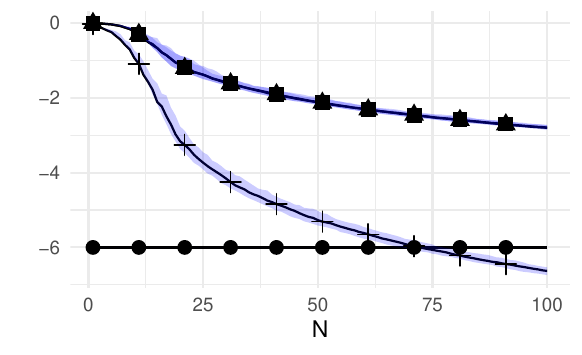}
\caption{$d = 2$, $\kappa = 0.1$.}
\label{fig:vmf-adapt-d1-kappa1}
\end{subfigure}
\begin{subfigure}{0.32\textwidth}
\includegraphics[width=\textwidth]{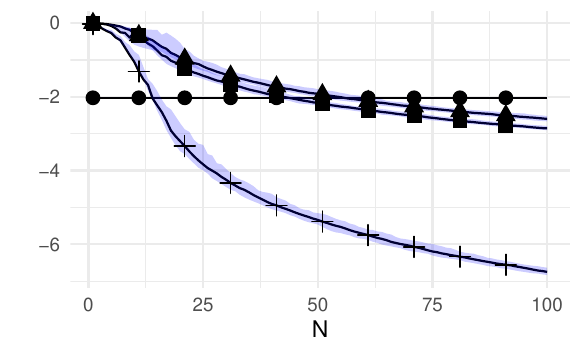}
\caption{$d = 2$, $\kappa = 1$.}
\label{fig:vmf-adapt-d1-kappa2}
\end{subfigure}
\begin{subfigure}{0.32\textwidth}
\includegraphics[width=\textwidth]{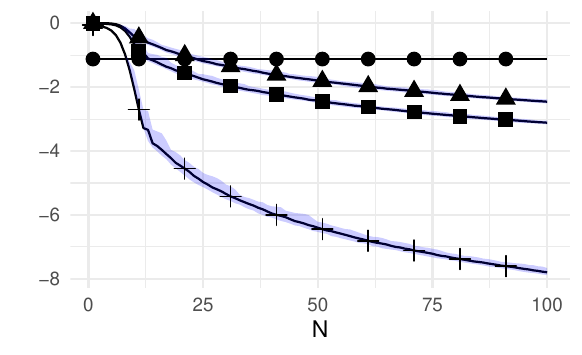}
\caption{$d = 2$, $\kappa = 10$.}
\label{fig:vmf-adapt-d1-kappa3}
\end{subfigure}

\begin{subfigure}{0.32\textwidth}
\includegraphics[width=\textwidth]{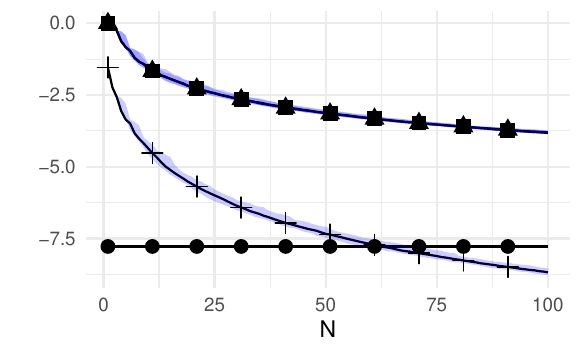}
\caption{$d = 4$, $\kappa = 0.1$.}
\label{fig:vmf-adapt-d2-kappa1}
\end{subfigure}
\begin{subfigure}{0.32\textwidth}
\includegraphics[width=\textwidth]{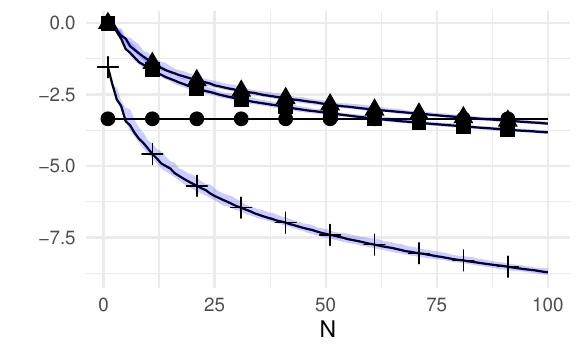}
\caption{$d = 4$, $\kappa = 1$.}
\label{fig:vmf-adapt-d2-kappa2}
\end{subfigure}
\begin{subfigure}{0.32\textwidth}
\includegraphics[width=\textwidth]{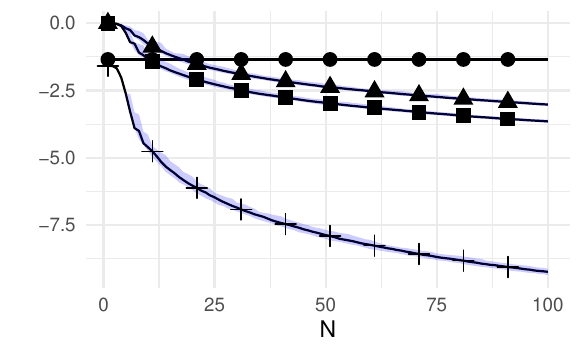}
\caption{$d = 4$, $\kappa = 10$.}
\label{fig:vmf-adapt-d2-kappa3}
\end{subfigure}

\begin{subfigure}{0.32\textwidth}
\includegraphics[width=\textwidth]{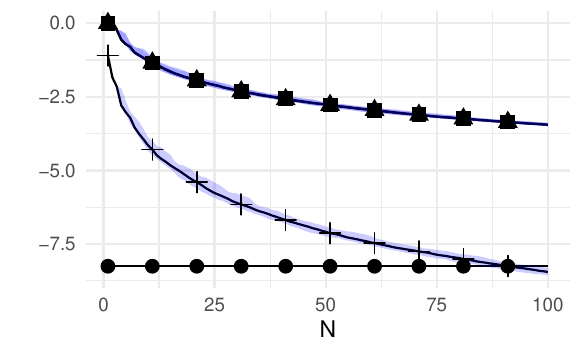}
\caption{$d = 5$, $\kappa = 0.1$.}
\label{fig:vmf-adapt-d3-kappa1}
\end{subfigure}
\begin{subfigure}{0.32\textwidth}
\includegraphics[width=\textwidth]{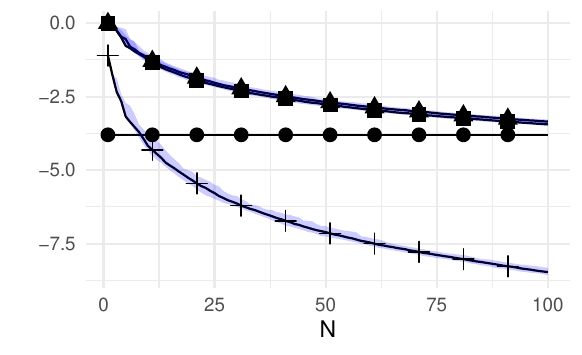}
\caption{$d = 5$, $\kappa = 1$.}
\label{fig:vmf-adapt-d3-kappa2}
\end{subfigure}
\begin{subfigure}{0.32\textwidth}
\includegraphics[width=\textwidth]{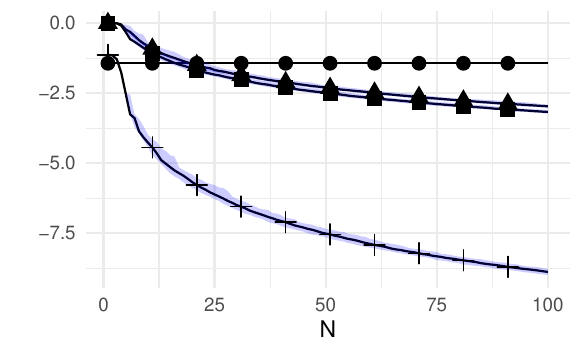}
\caption{$d = 5$, $\kappa = 10$.}
\label{fig:vmf-adapt-d3-kappa3}
\end{subfigure}
\caption{Log of rejection probability $\log(1 - \psi / \psi_N)$ using several samplers: UW ($\CIRCLE$), constant VS ($\blacktriangle$), constant VWS ($\blacksquare$), and linear VWS ($+$).}
\label{fig:vmf-adapt}
\end{figure}

\subsection{Approximate Computation of Probabilities}
\label{sec:vmf-probabilities}

Proposition~\ref{result:total-variation} established that proposal $h$ may be useful in approximating probabilities under $f$ when bound \eqref{eqn:bound} can be made small. To illustrate, let us consider the probability that $\vec{V}_0 \sim \text{VMF}_d(\vec{\mu}_0, \kappa)$ lies in the nonnegative orthant $A = \{ \vec{v} \in \mathbb{R}^d : \vec{v} \geq 0 \}$.
Using the transformation from $(X, \vec{U})$ to $\vec{V}_0$ described in Section~\ref{sec:vmf-generation},
\begin{align*}
\Prob(\vec{V}_0 \in A) &= 
\Prob\left( X \geq 0, U_1 \sqrt{1-X^2} \geq 0, \ldots, U_{d-1} \sqrt{1-X^2} \geq 0 \right) \\
&= \Prob\left( X \geq 0, U_1 \geq 0, \ldots, U_{d-1} \geq 0 \right) \\
&= \Prob(U_1 \geq 0, \ldots, U_{d-1} \geq 0) \Prob(X \geq 0) \\
&= 2^{-(d-1)} \Prob(X \geq 0).
\end{align*}
Similarly, let $\tilde{\vec{V}}_0 = (\tilde{X}, [1 - \tilde{X}^2]^{1/2} \cdot \vec{U}))$ with $\tilde{X} \sim h$ so that $\Prob(\tilde{\vec{V}}_0 \in A) = 2^{-(d-1)} \Prob(\tilde{X} \geq 0)$. The bound \eqref{eqn:bound} gives
\begin{align}
\Delta &:= |\Prob(\tilde{\vec{V}}_0 \in A) - \Prob(\vec{V}_0 \in A)| \nonumber \\
&= 2^{-(d-1)} |\Prob(\tilde{X} \geq 0) - \Prob(X \geq 0)| \nonumber \\
&\leq 2^{-(d-1)} (1 - \psi / \psi_N).
\label{eqn:vmf-orthant-bound}
\end{align}
A brief study in Appendix~\ref{sec:vmf-probabilities-study} compares $\Delta$ with the actual approximation error---for several values of $\kappa$---and $d$ and demonstrates their reduction as $N$ increases. The bound is seen to be conservative as may have been anticipated because it is not specific to the event $[X \geq 0]$.

\subsection{A Bayesian Application}
\label{sec:vmf-bayesian}

A third setting which can make use of vertical weighted strips is in a Bayesian analysis with independent and identically distributed VMF outcomes. \citet{DamienWalker1999} propose a full Bayesian treatment for circular data ($d = 2$) based on Gibbs sampling with data augmentation. \citet{NunezAntonioGutierrezPena2005} propose a Bayesian treatment for $d \geq 2$ based on sampling-importance-resampling. We will develop a rejection sampler based on vertical weighted strips for the VMF setting; although the posterior does not follow a familiar distribution, samples can be generated from it exactly without resorting to MCMC.

Suppose $\vec{v}_1, \ldots, \vec{v}_n$ are an independent and identically distributed sample from $\text{VMF}_d(\mu, \kappa)$ with unknown $\kappa > 0$ and $\vec{\mu} \in \mathbb{S}^d$. A conjugate prior in this setting is given by
\begin{align*}
\pi^{(0)}(\vec{\mu}, \kappa) &\propto
\left[ \frac{\kappa^{d/2 - 1}}{I_{d/2 - 1}(\kappa)} \right]^{c_0} \exp(\kappa R_0 \vec{m}_0^\top \vec{\mu}),
\end{align*}
where $c_0 \geq 0$, $R_0 \geq 0$, and $\vec{m}_0 \in \mathbb{S}^d$ \citep{MardiaElAtoum1976}. Upon observing $\vec{v}_1, \ldots, \vec{v}_n$, the posterior distribution for $[\vec{\mu}, \kappa \mid \vec{v}_1, \ldots, \vec{v}_n]$ is
\begin{align*}
\pi^{(n)}(\vec{\mu}, \kappa) &\propto
\pi^{(0)}(\vec{\mu}, \kappa) \left[ \frac{\kappa^{d/2 - 1}}{I_{d/2 - 1}(\kappa)} \right]^n
\exp\left\{
\kappa \vec{\mu}^\top \sum_{i=1}^n \vec{v}_i
\right\} \\
&= \left[ \frac{\kappa^{d/2 - 1}}{I_{d/2 - 1}(\kappa)} \right]^{c_0 + n}
\exp\left\{
\kappa R_n \vec{\mu}^\top \vec{m}_n
\right\},
\end{align*}
where $\vec{m}_n = R_n^{-1} (\sum_{i=1}^n \vec{v}_i + R_0 m_0)$ and $R_n$ is the Euclidean norm of $\sum_{i=1}^n \vec{v}_i + R_0 m_0$. Therefore, $\pi^{(0)}(\vec{\mu}, \kappa)$ is seen to be a conjugate prior. Notice that 
\begin{align*}
\pi^{(n)}(\vec{\mu}, \kappa) &\propto
\left[ \frac{\kappa^{d/2 - 1}}{I_{d/2 - 1}(\kappa)} \right]^{c_0 + n - 1}
\frac{I_{d/2 - 1}(\kappa R_n)}{I_{d/2 - 1}(\kappa)}
f_\text{VMF}(\vec{\mu} \mid \vec{m}_n, \kappa R_n)
\end{align*}
is the product of conditional distribution $[\vec{\mu} \mid \kappa, \vec{v}_1, \ldots \vec{v}_n] \sim \text{VMF}(\vec{\mu} \mid \vec{m}_n, \kappa R_n)$ and marginal $[\kappa \mid \vec{v}_1, \ldots \vec{v}_n]$ with unnormalized density
\begin{displaymath}
f_0(\kappa) = 
\left[ \kappa^{d/2 - 1} / I_{d/2 - 1}(\kappa) \right]^{c_0 + n - 1}
I_{d/2 - 1}(\kappa R_n) / I_{d/2 - 1}(\kappa).
\end{displaymath}
Therefore, exact generation of variates from the posterior may be accomplished by first drawing $\kappa$ from $f_0$ then $\vec{\mu}$ from $\text{VMF}(\vec{\mu} \mid \vec{m}_n, \kappa R_n)$. The latter has been explored in Section~\ref{sec:vmf-generation}, so we now focus on the target $f_0$. Consider the decomposition $f_0(\kappa) = w(\kappa) g(\kappa)$ with $g(\kappa) = \tau e^{ -\tau \kappa } \cdot \ind(\kappa > 0)$ the density of the exponential distribution with rate $\tau > 0$ of our choosing and weight function
\begin{align*}
w(\kappa) &= \tau^{-1} \cdot e^{ \tau \kappa }
\left[ \frac{\kappa^{d/2 - 1} }{ I_{d/2 - 1}(\kappa)} \right]^{c+n-1}
\frac{I_{d/2 - 1}(\kappa R_n)}{I_{d/2 - 1}(\kappa)} \\
&= \tau^{-1} \cdot e^{-\kappa (1 - R_n - \tau)}
\left[ \frac{\kappa^{d/2 - 1} }{ e^\kappa \mathcal{I}_{d/2 - 1}(\kappa)} \right]^{c+n-1}
\frac{\mathcal{I}_{d/2 - 1}(\kappa R_n)}{\mathcal{I}_{d/2 - 1}(\kappa)}.
\end{align*}
The exponentially scaled Bessel function $\mathcal{I}_{\nu}(x) = e^{-x} I_\nu(x)$, computed with \texttt{besselI} in R, is useful for working on the log-scale to avoid precision issues due to very large or very small magnitude numbers. To carry out VWS sampling from $f_0$, we opt for the constant majorizer described in Example~\ref{example:vws-constant}, using numerical optimization to find the minimum and maximum of $\log w(\kappa)$ on each region $\mathscr{D}_j$, and take $\tau = 0.01$ so that the proposal is not closely concentrated around zero. We will make use of bound \eqref{eqn:bound} and avoid computing the normalizing constant $\psi$ of $f_0$. 

Let us consider a dataset from Appendix~B2 of \citet{FisherLewisEmbleton1993} with $n = 26$ measurements of magnetic remanence in specimens of Palaeozoic red-beds from Argentina. Measurements are initially given as declination / inclination coordinates $(\theta_{i1}, \theta_{i2})$ in degrees and transformed to $\mathbb{R}^3$ using
\begin{align*}
v_{i1} = \sin(\vartheta_{i2}) \cos(\vartheta_{i1}), \quad
v_{i2} = \sin(\vartheta_{i2}) \sin(\vartheta_{i1}), \quad
v_{i3} = \cos(\vartheta_{i2}),
\end{align*}
where 
$\vartheta_{i1} = (360^\circ - \theta_{i1}) \pi / 180^\circ$ and $\vartheta_{i2} = (90^\circ + \theta_{i1}) \pi / 180^\circ$. A benefit of a rejection sampling approach such as VWS is that accepted draws will be an exact sample from the target posterior distribution. We take hyperparameters $c_0 = 0$ and $R_0 = 0$ to match Example~5.4 of \citet{NunezAntonioGutierrezPena2005}.

Figure~\ref{fig:vmf-bayesian} displays the results of rejection sampling. The bound \eqref{eqn:bound} shown in Figure~\ref{fig:vmf-bayesian-adapt} reduces to 11.4\% with $N = 50$ regions. Using this proposal, 6,363 rejections were encountered to obtain a sample of 100,000 (5.98\% rejection). The empirical density of the posterior based on the accepted draws is displayed in Figure~\ref{fig:vmf-bayesian-kappa-posterior}. An estimate of $\kappa$ based on the posterior mean is $\tilde{\kappa} = 113.24$ and an associated 95\% credible interval based on $2.5\%$ and $97.5\%$ quantiles is $(74.00, 160.72)$. For comparison, consider the MLE computed by numerical maximization of the log-likelihood; here, we transform from three Euclidean pre-parameters $\vec{\zeta} = (\zeta_1, \zeta_2, \zeta_3)$ to $(\kappa, \vec{\mu})$ to enforce $\kappa > 0$ and $\vec{\mu} \in \mathbb{S}^d$ with unconstrained numerical optimization. We obtain the estimate $\hat{\kappa} = 113.24$ and an associated 95\% confidence interval $(77.07, 166.24)$.

\begin{figure}
\centering
\begin{subfigure}{0.48\textwidth}
\includegraphics[width=\textwidth]{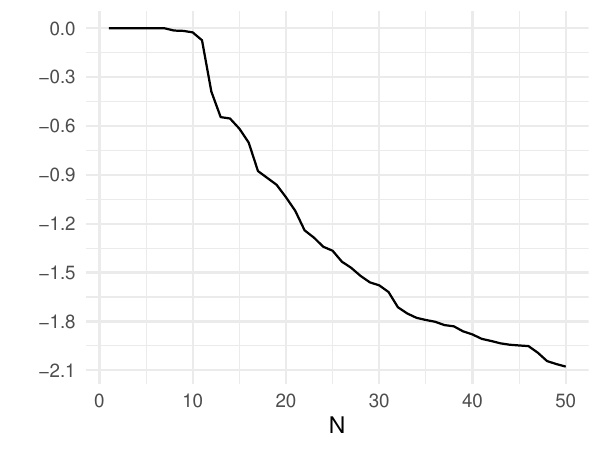}
\caption{}
\label{fig:vmf-bayesian-adapt}
\end{subfigure}
\begin{subfigure}{0.48\textwidth}
\includegraphics[width=\textwidth]{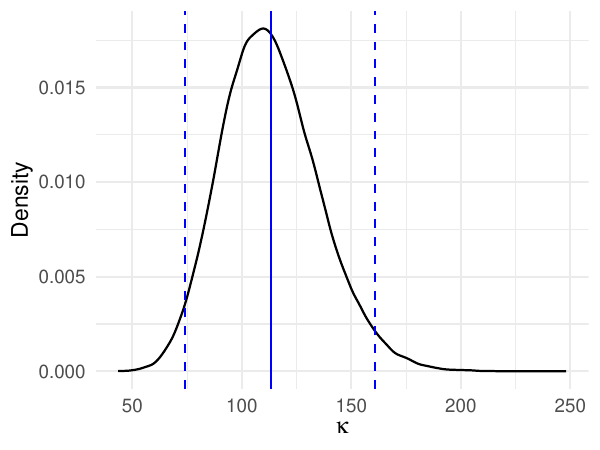}
\caption{}
\label{fig:vmf-bayesian-kappa-posterior}
\end{subfigure}
\caption{
Results for posterior $[\kappa \mid \vec{v}_1, \ldots \vec{v}_n]$. (\subref{fig:vmf-bayesian-adapt}) Log of bound \eqref{eqn:bound} for rejection probability as $N$ increases from 1 to 50. (\subref{fig:vmf-bayesian-kappa-posterior}) Empirical density of accepted draws (solid curve) with mean (solid vertical line) and 2.5\% and 97.5\% quantiles (dashed vertical lines).}
\label{fig:vmf-bayesian}
\end{figure}

\section{Conclusions}
\label{sec:conclusions}

This paper has explored vertical weighted strips (VWS), a generalization of the vertical strips method to construct proposals for rejection sampling. Regarding the target as a weighted density provides additional flexibility in constructing the proposal. This approach is effective when the portion designated as the weight function is majorized by a more convenient function, and this majorizer recombined with the remaining base density yields a distribution for which draws are conveniently generated. Several examples were given to demonstrate situations where practical samplers can be achieved. Highly efficient samplers were obtained in some cases using only a moderate number of regions. The framework provides insight into the rejection rate which may help to guide proposal construction.

We focused on two particular majorizers: one based on a constant and one based on a linear function on the logarithmic scale. A source of inspiration for other useful inequalities may be in the minorization-maximization (MM) literature \citep[e.g.][]{Lange2016}, where minorization is used to construct a sequence of surrogates to a complicated likelihood function which are more readily maximized to obtain an MLE.

The univariate setting of this paper most readily applies to multivariate sampling within the context of a Gibbs sampler. Here, VWS may be used to generate exact draws from unfamiliar univariate conditionals. There is a tradeoff between proposal construction time and sampling time: typically only one accepted draw is needed in each iteration of the Gibbs sampler so that a moderate rejection rate may be preferable to spending more time to craft the proposal each iteration.

The finite mixture of VMF densities is a relevant extension to the setting considered in Section~\ref{sec:vmf}. This model is useful both in clustering and flexible modeling of directional data. Bayesian fitting of such models is considered by \citet{TaghiaMaLeijon2014} and \citet{GopalYang2014}; both works consider a variational approximation to the posterior and the latter also proposes a Gibbs sampling method. Here, a collapsed Gibbs sampler is used to sequentially draw latent subject-specific class labels and class-specific concentration parameters; class-specific orientation parameters are assumed in the model but marginalized out so that they do not need to be drawn. While it is routine to draw class labels from a categorical distribution, the conditional of each concentration parameter has a less familiar weighted lognormal distribution. VWS may be used in place of the Metropolis step suggested in the work. Taking exact draws from the conditional should result in mixing better than (or as good as) a Metropolis step and avoids the issue of tuning to find a suitable step size. The amount of computation to obtain a proposal for each concentration parameter may be substantially larger than a decision rule in a typical Metropolis algorithm, but overall run time per Gibbs iteration should still be sensible for small-to-moderate numbers of classes typically used with finite mixture models.

There is also potential for VWS methodology to be applied directly to multivariate settings. Rather than intervals which have been used in the univariate case, it may be necessary to partition along multiple dimensions---e.g., with hyperrectangles---in such settings. Generation of proposed draws from subsequent reweighted and truncated base distributions must then be practical for a usable sampler. This approach appears viable for some problems---in lower dimensional settings or where special structure exists---and may be an interesting area for future work.


\appendix

\section{Proofs of Propositions}
\label{sec:proofs}

This section provides proofs of Propositions~\ref{result:rejection-probability-bound-fmm} and
Proposition~\ref{result:total-variation}.

\begin{proof}[Proof of Proposition~\ref{result:rejection-probability-bound-fmm}]
The true rejection probability is
\begin{align*}
\frac{\psi_N - \psi}{\psi_N}
&= \frac{1}{\psi_N} \int_\Omega [h_0(x) - f_0(x)] d\nu(x) \nonumber \\
&= \frac{1}{\psi_N} \sum_{j=1}^N \int_\Omega \ind(x \in \mathscr{D}_j) [\overline{w}_j(x) - w(x)] g(x) d\nu(x) \nonumber \\
&\leq \frac{1}{\psi_N} \sum_{j=1}^N \int_\Omega \ind(x \in \mathscr{D}_j) [\overline{w}_j(x) - \underline{w}_j(x)] g(x) d\nu(x) \nonumber \\
&= \frac{1}{\psi_N} \left\{ \sum_{j=1}^N \overline{\xi}_j - \sum_{j=1}^N \underline{\xi}_j \right\},
\end{align*}
which is equivalent to \eqref{eqn:bound}.
\end{proof}

\begin{proof}[Proof of Proposition~\ref{result:total-variation}]
For any $B \in \mathcal{B}$,
\begin{align}
\int_B h(x) d\nu(x) - \int_B f(x) d\nu(x)
&= \frac{1}{\psi_N} \int_B h_0(x) d\nu(x) - \frac{1}{\psi} \int_B f_0(x) d\nu(x) \nonumber \\
&\leq \frac{1}{\psi_N} \left[ \int_B h_0(x) d\nu(x) - \int_B f_0(x) d\nu(x) \right] \nonumber \\
&= \frac{1}{\psi_N}
\sum_{j=1}^N \int_{B \cap \mathscr{D}_j} [\overline{w}_j(x) - w(x)] g(x) d\nu(x) \nonumber \\
&\leq \frac{1}{\psi_N}
\sum_{j=1}^N \int_{\mathscr{D}_j} [\overline{w}_j(x) - w(x)] g(x) d\nu(x) \nonumber \\
&= \frac{\psi_N - \psi}{\psi_N}
\label{eqn:norm-diff1}
\end{align}
and
\begin{align}
\int_B f(x) d\nu(x) - \int_B h(x) d\nu(x)
&= \frac{1}{\psi} \int_B f_0(x) d\nu(x) - \frac{1}{\psi_N} \int_B h_0(x) d\nu(x) \nonumber \\
&\leq \frac{1}{\psi} \int_B f_0(x) d\nu(x) - \frac{1}{\psi_N} \int_B f_0(x) d\nu(x) \nonumber \\
&= \frac{\psi_N - \psi}{\psi_N} \int_B f(x) d\nu(x) \nonumber \\
&\leq \frac{\psi_N - \psi}{\psi_N}.
\label{eqn:norm-diff2}
\end{align}
Combining \eqref{eqn:norm-diff1} and \eqref{eqn:norm-diff2} gives the result.
\end{proof}

\section{Sampling from a Univariate Proposal}
\label{sec:sampling-univariate-proposal}
The cumulative distribution function (CDF) associated with proposal distribution $h$ may be written as
\begin{align}
H(x) &= \sum_{j=1}^{\ell-1} \pi_j + \pi_\ell G_\ell(x), \quad
G_{\ell}(x) = \frac{\int_{(\alpha_{\ell-1}, x]} \overline{w}_\ell(s) g(s) d\nu(s)}{ \int_{(\alpha_{\ell-1}, \alpha_\ell]} \overline{w}_\ell(s) g(s) d\nu(s)}, \quad
\text{if $x \in (\alpha_{\ell-1}, \alpha_{\ell}]$}.
\label{eqn:cdf-quantile}
\end{align}
To obtain the quantile function, suppose $\varphi \in [0,1]$ is a desired quantile and $\ell$ is the index such that $\sum_{j=1}^{\ell-1} \pi_j < \varphi \leq \sum_{j=1}^{\ell} \pi_j$; then
\begin{align}
&\varphi \leq H(x) \equiv \sum_{j=1}^{\ell-1} \pi_j + \pi_\ell G_\ell(x) \nonumber \\
&\iff \quad
G_\ell(x) \geq 
\frac{
\varphi - \sum_{j=1}^{\ell-1} \pi_j}{
\pi_\ell
} = 
\frac{1}{\overline{\xi}_\ell} \left[ \varphi \sum_{j=1}^{N} \overline{\xi}_j -
\sum_{j=1}^{\ell-1} \overline{\xi}_j \right].
\label{eqn:prep-quantile}
\end{align}
Therefore, the quantile function associated with $H$ is
\begin{align}
H^{-}(\varphi) &= \inf \{ x \in \Omega : H(x) \geq \varphi \} \nonumber \\
&= \inf \left\{ x \in \Omega : G_\ell(x) \geq \frac{1}{\overline{\xi}_\ell} \left[ \varphi \sum_{j=1}^{N} \overline{\xi}_j - \sum_{j=1}^{\ell-1} \overline{\xi}_j \right] \right\} \nonumber \\
&= G_\ell^{-}\left(
\frac{1}{\overline{\xi}_\ell} \left[ \varphi \sum_{j=1}^{N} \overline{\xi}_j - \sum_{j=1}^{\ell-1} \overline{\xi}_j \right]
\right).
\label{eqn:quantile}
\end{align}
Computations involving $H$ and $H^{-}$ can be facilitated by precomputing the cumulative sums
\(
H(x_\ell) = \sum_{j=1}^{\ell} \pi_j
\)
for $\ell = 1, \ldots, N$. Then, for example, a binary search can be carried out to find the smallest $\ell$ such that $H(x_\ell) \geq \varphi$. Variates from $h$ can be generated with the inverse CDF method as $x = H^{-}(u)$ where $u$ is a draw from $\text{Uniform}(0,1)$.

In the case of the constant majorizer discussed in Example~\ref{example:vws-constant}, the CDF in \eqref{eqn:cdf-quantile} simplifies to
\begin{align*}
G_{\ell}(x) = \frac{G(x) - G(\alpha_{\ell-1})}{G(\alpha_\ell) - G(\alpha_{\ell-1})},
\quad \text{if $x \in (\alpha_{\ell-1}, \alpha_{\ell}]$}
\end{align*}
and \eqref{eqn:prep-quantile} simplifies to
\begin{align*}
&G_\ell(x) \geq
\frac{1}{\overline{\xi}_\ell} \left[ \varphi \sum_{j=1}^{N} \overline{\xi}_j -
\sum_{j=1}^{\ell-1} \overline{\xi}_j \right] \\
&\iff \frac{G(x) - G(\alpha_{\ell-1})}{G(x_\ell) - G(\alpha_{\ell-1})} \geq
\frac{
\varphi \sum_{j=1}^{N} \overline{w}_j [G(\alpha_j) - G(\alpha_{j-1})] -
\sum_{j=1}^{\ell-1} \overline{w}_j [G(\alpha_j) - G(\alpha_{j-1})]
}{
\overline{w}_\ell [G(\alpha_\ell) - G(\alpha_{\ell-1})]
} \\
&\iff G(x) \geq G(\alpha_{\ell-1}) +
\frac{1}{\overline{w}_\ell} \left[
\varphi \sum_{j=1}^{N} \overline{w}_j [G(\alpha_j) - G(\alpha_{j-1})] -
\sum_{j=1}^{\ell-1} \overline{w}_j [G(\alpha_j) - G(\alpha_{j-1})]
\right].
\end{align*}
Therefore, the expression for the $\varphi$ quantile of $H$ in \eqref{eqn:quantile} can be expressed as a quantile of base distribution $G$.

Section~\ref{sec:design} discussed several conditions at extreme endpoints of $\Omega = (\alpha_0, \alpha_N]$ which can complicate use of the constant or linear majorizers. A simple workaround involves truncating $\Omega$ to exclude such endpoints; however, the following remark briefly discusses another approach which can be used to avoid such truncation. 

\begin{remark}
\label{remark:infinite-log-weight}
Suppose region $\mathscr{D}_1 = (\alpha_0, \alpha_1]$ has an infinite value of $w(\alpha_0)$---i.e., at the left endpoint of the support---so that the constant and linear majorizers of Example~\ref{example:vws-constant} and Section~\ref{sec:vws-linear}, respectively, cannot be used with it. Rejection sampling can be modified so that an immediate rejection occurs if $\mathscr{D}_1$ is selected when drawing from $h$. In particular, define
\begin{displaymath}
h^{(-1)}(x) = \sum_{j=2}^N \pi_j^{(-1)} g_j(x), \quad
h_0^{(-1)}(x) = \sum_{j=2}^N \overline{w}_j(x) g(x) \ind\{ x \in \mathcal{D}_j \}, \quad
\psi_N^{(-1)} = \sum_{j=2}^N \overline{\xi}_j, \quad 
\pi_j^{(-1)} = \frac{\overline{\xi}_j}{\psi_N^{(-1)}},
\end{displaymath}
for $j = 2, \ldots, N$. A two stage rejection sampling algorithm first draws $\ell$ from $1, \ldots, N$ with probabilities $\pi_1, \ldots, \pi_N$ based on $h$. A rejection occurs at the first stage if $\ell = 1$. In this case, region $\mathscr{D}_1$ may be partitioned, say, into
$\mathscr{D}_1^{(1)} = (\alpha_0, \alpha_*]$ and
$\mathscr{D}_1^{(2)} = (\alpha_*, \alpha_1]$,
so that a constant or linear majorizer can be used on $\mathscr{D}_1^{(2)}$ in subsequent attempts. If $\ell \in \{ 2, \ldots, N\}$, $u$ and $x$ are drawn from $\text{Uniform}(0,1)$ and $h^{(-1)}$, respectively. The second stage accepts $x$ as a draw from target $f$ if $u \leq f_0(x) / h_0^{(-1)}(x)$. The rejection probability over both stages is
\begin{align}
\Prob(\text{Reject}) &= \Prob(\ell = 1) +
\Prob\Big(\ell > 1, U > f_0(X) / h_0^{(-1)}(X) \Big) \nonumber \\
&= \Prob(\ell = 1) +
\Prob\Big(U > f_0(X) / h_0^{(-1)}(X) \mid \ell > 1 \Big) 
\Prob(\ell > 1) \nonumber \\
&= \pi_1 +
\left( 1 - \pi_1 \right)
\left( 1 - \psi / \psi_N^{(-1)} \right). \label{eqn:composite}
\end{align}
To define $\overline{\xi}_1$ needed in the first stage to draw $\ell$, we may assume the trivial majorizer from Remark~\ref{remark:trivial-minorizer} so that $\overline{\xi}_1 = \int_{\mathscr{D}_\ell} w(x) g(x) d\nu(x)$. Bound \eqref{eqn:bound} may be applied to \eqref{eqn:composite} to obtain 
\begin{align*}
\Prob(\text{Reject}) \leq
\pi_1 +
\left( 1 - \pi_1 \right)
\left( 1 - \frac{1}{\psi_N^{(-1)}} \sum_{j=2}^N \underline{\xi}_j \right).
\end{align*}
\end{remark}

\section{Details for Constant Majorizer}
\label{sec:constant-majorizer-details}

We make use of several standard numerical optimization methods to majorize and minorize $w$. Consider maximizing $w$ on the interval $(a, b]$. When both endpoints are finite, Brent's method \citep{Brent1973} is used with $a$ and $b$ as bounds. Otherwise, the BFGS quasi-Newton method \citep[Section~6]{NocedalWright2006} is used to maximize $w(t(z))$ with respect to $z \in \mathbb{R}$; the transformation
\begin{align*}
t(z) =
\begin{cases}
z, & \text{if $a = -\infty$ and $b = \infty$}, \\
a + \exp(z), & \text{if $a > \infty$ and $b = \infty$}, \\
b - \exp(-z), & \text{if $a = -\infty$ and $b < \infty$},
\end{cases}
\end{align*}
ensures that $t(z) \in (a,b)$. The function $w$ is also checked explicitly at endpoints $a$ and $b$, as the maximum may occur at those points with an infinite value or which is otherwise not a critical point. A similar combination of methods is used in numerical minimization.

\section{Details for Linear Majorizer}
\label{sec:linear-majorizer-details}

The expansion point $c$ described in Section~\ref{sec:vws-linear} may be chosen to yield a small upper bound \eqref{eqn:linear-majorizer} over $x \in \mathscr{D}_j$ in some sense. The criterion we use for a majorizer in the log-concave case is to minimize the L1 distance between unnormalized densities $h_0$ and $f_0$ on $\mathscr{D}_j$,
\begin{align}
c^* &= \argmin_{c \in \mathscr{D}_j} \int_{\mathscr{D}_j} | h_0(x) - f_0(x) | d\nu(x) \nonumber \\
&= \argmin_{c \in \mathscr{D}_j} \int_{\mathscr{D}_j} [\overline{w}_j(x) - w(x) ] g(x) d\nu(x) \nonumber \\
&= \argmin_{c \in \mathscr{D}_j} \Big\{ w(c) \exp\{ -c \nabla(c) \} \int_{\mathscr{D}_j} \exp\{ x \nabla(c) \} g(x) d\nu(x) \Big\} \nonumber \\
&= \argmin_{c \in \mathscr{D}_j} \Big\{ w(c) \exp\{ -c \nabla(c) \} \Prob(T \in \mathscr{D}_j) M_j(\nabla(c)) \Big\} \nonumber \\
&= \argmin_{c \in \mathscr{D}_j} \Big\{ \log w(c) - c \nabla(c) + \log M_j(\nabla(c)) \Big\},
\label{eqn:concave-majorize-choice}
\end{align}
where $M_j(s) = \int_a^b e^{xs} g(x) d\nu(x)$ is the moment generating function of random variable $T$ whose density is $g$ with support truncated to $(\alpha_{j-1}, \alpha_j]$. Similar to \eqref{eqn:concave-majorize-choice}, a choice of $c$ for the minorizer in the log-convex case can be obtained from
\begin{align}
c^* &= \argmin_{c \in \mathscr{D}_j} \int_{\mathscr{D}_j} [w(x) - \underline{w}_j(x) ] g(x) d\nu(x) \nonumber \\
&= \argmax_{c \in \mathscr{D}_j} \int_{\mathscr{D}_j} \underline{w}_j(x) g(x) d\nu(x) \nonumber \\
&= \argmax_{c \in \mathscr{D}_j} \Big\{ \log w(c) - c \nabla(c) + \log M_j(\nabla(c)) \Big\}.
\label{eqn:convex-minorize-choice}
\end{align}
Criteria \eqref{eqn:concave-majorize-choice} and \eqref{eqn:convex-minorize-choice} are utilized in Section~\ref{sec:vmf} where optimization is carried out numerically via Brent's method \citep{Brent1973} within a bounded support; expressions for the function $M_j(s)$ are given by 
\begin{align*}
M_j(s)
= \frac{\kappa}{s + \kappa} \frac{
e^{ (s + \kappa) \alpha_j } - e^{ (s + \kappa) \alpha_{j-1} }
}{
e^{\kappa \alpha_j} - e^{\kappa \alpha_{j-1}}
}
\end{align*}
for the setting in Example~\ref{example:doubly-truncated-exp} and 
\begin{align*}
M_j(s)
= \frac{
e^{ s \alpha_j } - e^{ s \alpha_{j-1} }
}{
s (\alpha_j - \alpha_{j-1})
}
\end{align*}
for the setting in Example~\ref{example:uniform}.

\section{Additional Results for VMF Applications}
\label{sec:vmf-details}

\subsection{Generation of Variates}
\label{sec:vmf-generation-extra}

Several additional displays are given here to accompany the discussion in Section~\ref{sec:vmf-generation}. Figure~\ref{fig:vmf-simple-reweighted} compares the simple proposal of Example~\ref{example:simple-proposal} with the target for the case $\kappa = 10$ and $d = 10$; it is apparent that $h_0$ is not an efficient majorizer for $f_0$ as $x$ increases beyond 0.5. Figure~\ref{fig:vmf-target} displays several cases of the target density $f$ with $\kappa \in \{ 0.1, 10 \}$ and $d \in \{ 2, 3, 4 \}$. Figure~\ref{fig:vmf-proj2d} displays 50,000 draws of a three-dimensional VMF distribution with $\kappa \in \{ 0.1, 10 \}$ constructed from variates from rejection sampling on $f$.

In this setting, it is convenient to construct a VS proposal with a linear majorizer. This is formulated in the following example.

\begin{example}[Linear VS]
\label{example:vmf-linear-vs}
Consider the decomposition of $f$ in Example~\ref{example:vmf-constant-vs}. Here,
\begin{math}
\frac{d^2}{dx^2} \log w(x) =  -(d-3) \frac{1 + x^2}{(1 - x^2)^2}
\end{math}
so that $w$ is log-convex, log-concave, or constant under the same conditions in Example~\ref{example:vmf-linear-vws}. Example~\ref{example:uniform} shows that mixture component $g_j$ in the proposal corresponds to distribution $\text{Exp}_{(\alpha_{j-1}, \alpha_j)}(\overline{\beta}_{j1})$ and gives expressions for $\overline{\xi}_j$ and $\underline{\xi}_j$. As in Example~\ref{example:vmf-linear-vws}, we may instead opt to use the trivial minorizer to compute $\underline{\xi}_j$ and obtain a tighter bound in \eqref{eqn:bound}.
\end{example}

The following result shows that the proposals in Examples~\ref{example:vmf-linear-vws} and \ref{example:vmf-linear-vs} are equivalent in practice. This equivalence is specific to the present setting where the two weight functions differ only by a linear term.

\begin{proposition}
\label{prop:linear-vmf}
The linear VWS and linear VS constructions given in Examples~\ref{example:vmf-linear-vws} and \ref{example:vmf-linear-vs}, respectively, yield the same acceptance ratios and the same rejection probabilities.
\end{proposition}

\begin{proof}
Let us write the coefficients in the present majorizer as $(\overline{\zeta}_{j0}, \overline{\zeta}_{j1})$ and the coefficients in the majorizer in Example~\ref{example:vmf-linear-vws} as $(\overline{\beta}_{j0}, \overline{\beta}_{j1})$. Let us first show that $\overline{\zeta}_{j0} \equiv \overline{\beta}_{j0}$ and $\overline{\zeta}_{j1} \equiv \overline{\beta}_{j1} + \kappa$. When $w$ is log-concave on $\mathcal{D}_j$,
\begin{align*}
\overline{\zeta}_{j1}
= -(d-3) \frac{c_j}{1 - c_j^2} + \kappa
= \overline{\beta}_{j1} + \kappa
\end{align*}
and
\begin{align*}
\overline{\zeta}_{j0} &= \log w(c_j) - c_j \overline{\zeta}_{j1} \\
&= \frac{d-3}{2} \log (1 - c_j^2) + \kappa c_j - c_j (\overline{\beta}_{j1} + \kappa) \\
&= \frac{d-3}{2} \log (1 - c_j^2) - c_j \overline{\beta}_{j1} \\
&= \overline{\beta}_{j0}.
\end{align*}
When $w$ is log-convex on $\mathcal{D}_j$,
\begin{align*}
\overline{\zeta}_{j1} &= \frac{
\log w(\alpha_j) - \log w(\alpha_{j-1})
}{
\alpha_j - \alpha_{j-1}
} \\
&= \frac{1}{\alpha_j - \alpha_{j-1}} 
\left[
\frac{d-3}{2} \log(1 - \alpha_j^2) + \kappa \alpha_j
- \frac{d-3}{2} \log(1 - \alpha_{j-1}^2) - \kappa \alpha_{j-1}
\right] \\
&= \overline{\beta}_{j1} + \kappa
\end{align*}
and
\begin{align*}
\overline{\zeta}_{j0}
&= \log w(\alpha_{j-1}) - \alpha_{j-1} \overline{\zeta}_{j1} \\
&= \frac{d-3}{2} \log(1 - \alpha_{j-1}^2) + \kappa \alpha_{j-1} - \alpha_{j-1} (\overline{\beta}_{j1} + \kappa) \\
&= \frac{d-3}{2} \log(1 - \alpha_{j-1}^2) - \alpha_{j-1} \overline{\beta}_{j1} \\
&= \overline{\beta}_{j0}.
\end{align*}
The acceptance ratio under Linear VS,
\begin{align*}
\frac{f_0(x)}{h_0(x)} &= \frac{
(1 - x^2)^{\frac{d-3}{2}} e^{\kappa x} \frac{1}{\alpha_N - \alpha_0}
\ind(\alpha_0 < x < \alpha_N)
}{
\sum_{j=1}^N \exp\{ \overline{\zeta}_{j0} + \overline{\zeta}_{j1} x \}
\frac{1}{\alpha_N - \alpha_0} \ind(\alpha_0 < x < \alpha_N)
} \\
&= \frac{
(1 - x^2)^{\frac{d-3}{2}}
}{
\sum_{j=1}^N \exp\{ \overline{\beta}_{j0} + \overline{\beta}_{j1} x \}
},
\end{align*}
is now seen to be equivalent to that of Linear VWS. Furthermore, under Linear VS,
\begin{align*}
\overline{\xi}_j
&= \frac{
\exp\{ \overline{\zeta}_{j0} \}
}{
(\alpha_N - \alpha_0) \overline{\zeta}_{j1}
}
\Big(
\exp\{ \overline{\zeta}_{j1} \cdot \alpha_j \} -
\exp\{ \overline{\zeta}_{j1} \cdot \alpha_{j-1} \}
\Big), \\
&= \frac{
\exp\{ \overline{\beta}_{j0} \}
}{
(\alpha_N - \alpha_0) (\overline{\beta}_{j1} + \kappa)
}
\Big(
\exp\{ (\overline{\beta}_{j1} + \kappa) \cdot \alpha_j \} -
\exp\{ (\overline{\beta}_{j1} + \kappa) \cdot \alpha_{j-1} \}
\Big),
\end{align*}
and 
\begin{align*}
\psi = \frac{f_0(x)}{f(x)}
= \frac{
\sqrt{\pi} \cdot I_{d/2-1}(\kappa) \cdot
\Gamma((d-1)/2)
}{
(\kappa / 2)^{d/2 - 1} (\alpha_N - \alpha_0)
}
\end{align*}
yields the acceptance probability
\begin{align}
\frac{\psi}{\psi_N} &=
\frac{
\sqrt{\pi} \cdot I_{d/2-1}(\kappa) \cdot \Gamma((d-1)/2)
}{
(\kappa / 2)^{d/2 - 1} (\alpha_N - \alpha_0)
}
\frac{1}{\sum_{j=1}^N \overline{\xi}_j} \nonumber \\
&= \frac{
\sqrt{\pi} \cdot I_{d/2-1}(\kappa) \cdot \Gamma((d-1)/2) (\kappa / 2)^{1 - d/2}
}{\sum_{j=1}^N 
\frac{
\exp\{ \overline{\beta}_{j0} \}
}{
(\overline{\beta}_{j1} + \kappa)
}
\Big(
\exp\{ (\overline{\beta}_{j1} + \kappa) \cdot \alpha_j \} -
\exp\{ (\overline{\beta}_{j1} + \kappa) \cdot \alpha_{j-1} \}
\Big)
}.
\label{eqn:linear-vs-acceptance}
\end{align}
The expression \eqref{eqn:linear-vs-acceptance} is also obtained under Linear VWS using $\overline{\xi}_j$ given in Example~\ref{example:doubly-truncated-exp} and $\psi$ given in \eqref{eqn:simple-norm-const}.
\end{proof}

Figure~\ref{fig:vmf-proposal} displays each of the four mixture proposals for $N = 5$ regions with $d = 2$ and $\kappa = 0.75$ along with the unnormalized target density. The interior knots $\alpha_1, \ldots, \alpha_4$ have been selected with equal spacing in all cases to aid the display. Overall rejection rates $1 - \psi / \psi_N$ of the constant VS, constant VWS, and linear VWS samplers are 97.43\%, 96.84\%, and 80.72\%, respectively. The contribution $\rho_\ell$ is displayed within each region. After refining to only $N = 5$ regions, it is apparent that linear VWS achieves a notably better rejection rate than the other two variations.

\begin{figure}
\centering
\begin{subfigure}{0.40\textwidth}
\includegraphics[width=\textwidth]{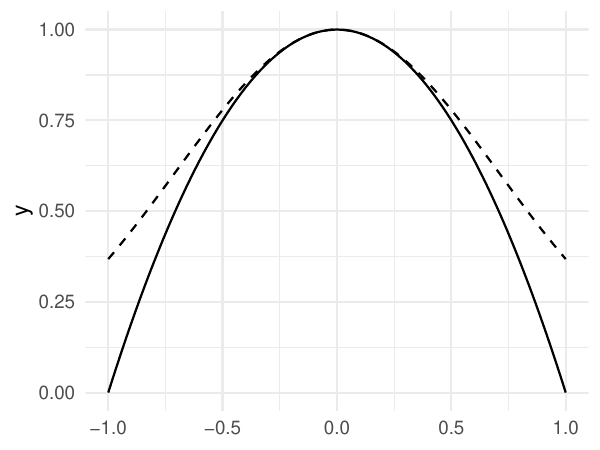}
\caption{}
\label{fig:vmf-simple-reweighted-weight-fns}
\end{subfigure}
\begin{subfigure}{0.40\textwidth}
\includegraphics[width=\textwidth]{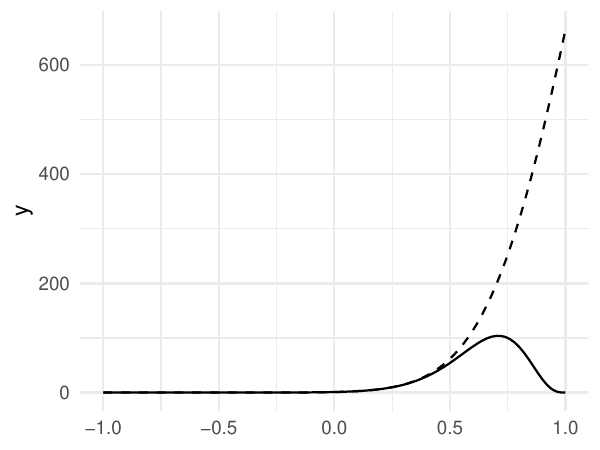}
\caption{}
\label{fig:vmf-simple-reweighted-densities}
\end{subfigure}
\caption{Simple proposal based on one region with $\kappa = 10$ and $d = 10$: (\subref{fig:vmf-simple-reweighted-weight-fns}) displays weight functions $w$ (solid) and $\overline{w}$ (dashed); (\subref{fig:vmf-simple-reweighted-densities}) displays unnormalized densities $f_0$ (solid) and $h_0$ (dashed).}
\label{fig:vmf-simple-reweighted}
\end{figure}

\begin{figure}
\centering
\begin{subfigure}{0.40\textwidth}
\includegraphics[width=\textwidth]{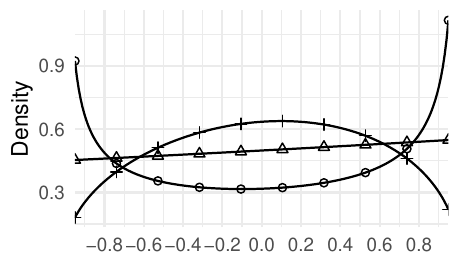}
\caption{$\kappa = 0.1$.}
\label{fig:vmf-target-kappa1}
\end{subfigure}
\begin{subfigure}{0.40\textwidth}
\includegraphics[width=\textwidth]{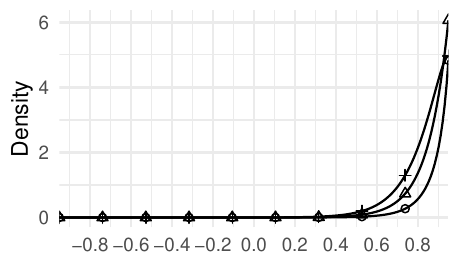}
\caption{$\kappa = 10$.}
\label{fig:vmf-target-kappa3}
\end{subfigure}
\caption{Density \eqref{eqn:vmf-target} used to draw from the VMF distribution: $d = 2$ ($\circ$), $d = 3$ ($\bigtriangleup$), and $d = 4$ ($+$).}
\label{fig:vmf-target}
\end{figure}

\begin{figure}
\centering
\begin{subfigure}{0.40\textwidth}
\includegraphics[width=\textwidth]{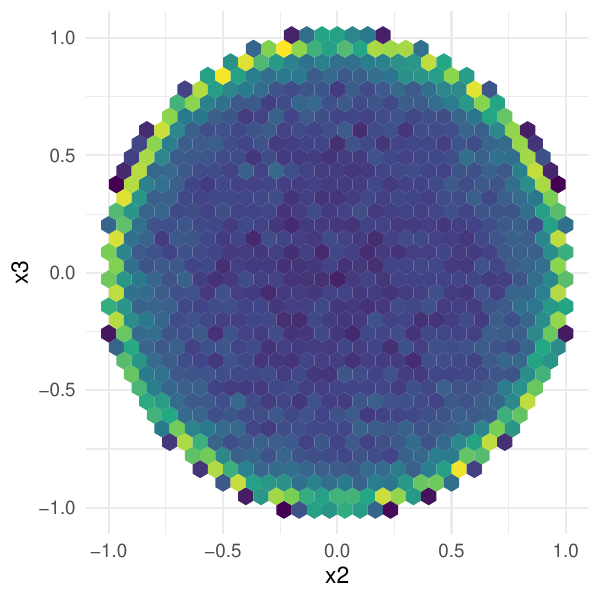}
\caption{$\kappa = 0.1$.}
\label{fig:vmf-proj2d-kappa1}
\end{subfigure}
\begin{subfigure}{0.40\textwidth}
\includegraphics[width=\textwidth]{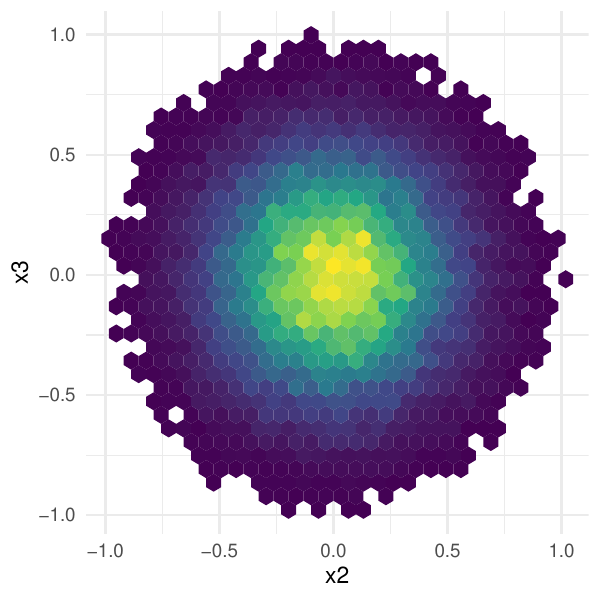}
\caption{$\kappa = 10$.}
\label{fig:vmf-proj2d-kappa3}
\end{subfigure}
\caption{Empirical distribution of 50,000 draws of $\vec{V}_0 \sim \text{VMF}_3(\vec{\mu}_0, \kappa)$, projected to the $x_2$--$x_3$ plane from $\vec{\mu}_0 = (1, 0, 0)$. Yellow bins contain a larger number of points while purple bins contain fewer points.}
\label{fig:vmf-proj2d}
\end{figure}

\begin{figure}
\centering
\begin{subfigure}{0.48\textwidth}
\includegraphics[width=\textwidth]{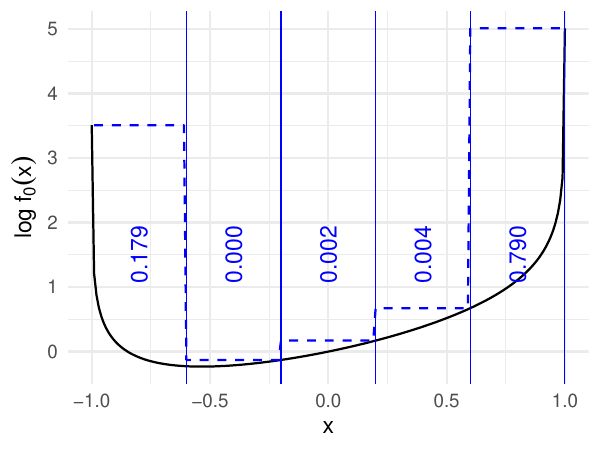}
\caption{VS.}
\label{fig:vmf-proposal-vs-const}
\end{subfigure}
\begin{subfigure}{0.48\textwidth}
\includegraphics[width=\textwidth]{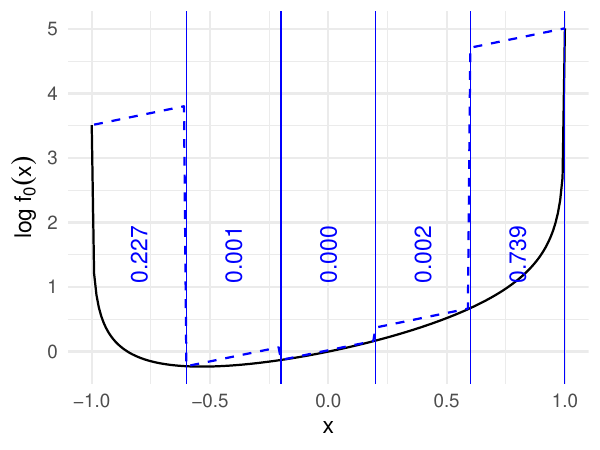}
\caption{Constant VWS.}
\label{fig:vmf-proposal-vws-const}
\end{subfigure}
\begin{subfigure}{0.48\textwidth}
\includegraphics[width=\textwidth]{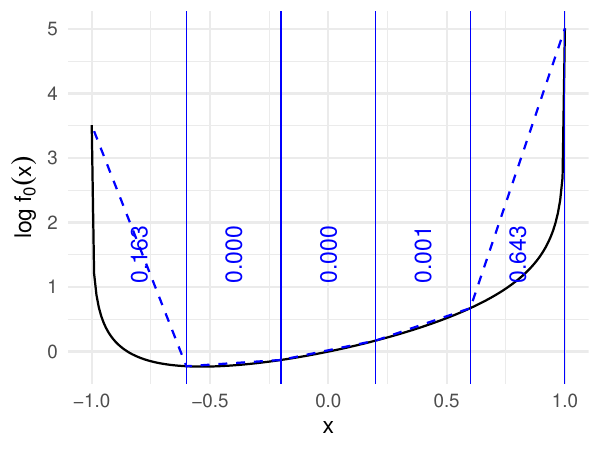}
\caption{Linear VWS.}
\label{fig:vmf-proposal-vws-linear}
\end{subfigure}
\caption{Unnormalized proposal log-density $\log h_0(x)$ for three VWS proposals (dashed blue curves) with $d = 2$, $\kappa = 0.75$ and $N = 5$ regions, and target log-density $\log f_0(x)$ (solid black curve). Solid horizontal blue lines are locations of interior knots $\alpha_1, \ldots, \alpha_4$. The value displayed within a region is its contribution $\rho_\ell$ to the rejection rate.}
\label{fig:vmf-proposal}
\end{figure}

\subsection{Approximate Computation of Probabilities Study}
\label{sec:vmf-probabilities-study}

This section presents a brief study of the approximation error and the upper bound obtained in Section~\ref{sec:vmf-probabilities}. Let us again consider the linear VWS proposal from Example~\ref{example:vmf-linear-vws}. Figure~\ref{fig:vmf-orthant} displays the realized approximation error $\Delta$ and its upper bound \eqref{eqn:vmf-orthant-bound}, on the log-scale, for $\kappa \in \{ 0.3, 1, 3\}$ and $d \in \{ 2, 4, 5 \}$ as $N$ increases from 1 to 100. The realized approximation error is often significantly smaller than the bound. Proposal $h$ is adapted using Algorithm~\ref{alg:refine} for each increment of $N$. Increases in $\Delta$ from increasing $N$ are possible due to refinements in $h$ which occur outside of event $A$. Such refinements make the error within $A$ relatively larger and might be avoided if $\Prob(\vec{V}_0 \in A)$ is the only aspect of $f$ of interest. With $N = 100$ regions, the largest error of the nine settings of $\kappa$ and $d$ is $\Delta = \exp(-8.754) \approx 1.58 \times 10^{-4}$ with $d = 2$ and $\kappa = 1$; for comparison, the value of the probability here is $\Prob(\vec{V}_0 \in A) \approx 0.3902$.

\begin{figure}
\centering
\begin{subfigure}{0.32\textwidth}
\includegraphics[width=\textwidth]{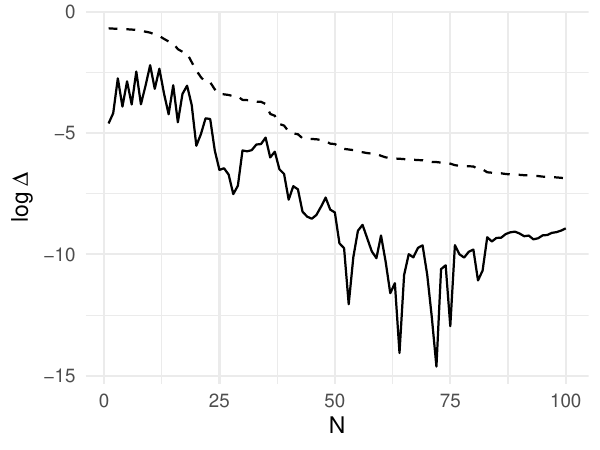}
\caption{$d = 2$, $\kappa = 0.3$.}
\label{fig:vmf-orthant-d1-kappa1}
\end{subfigure}
\begin{subfigure}{0.32\textwidth}
\includegraphics[width=\textwidth]{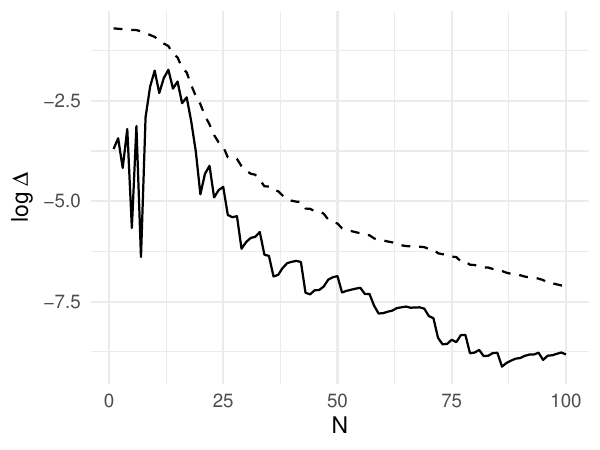}
\caption{$d = 2$, $\kappa = 1$.}
\label{fig:vmf-orthant-d1-kappa2}
\end{subfigure}
\begin{subfigure}{0.32\textwidth}
\includegraphics[width=\textwidth]{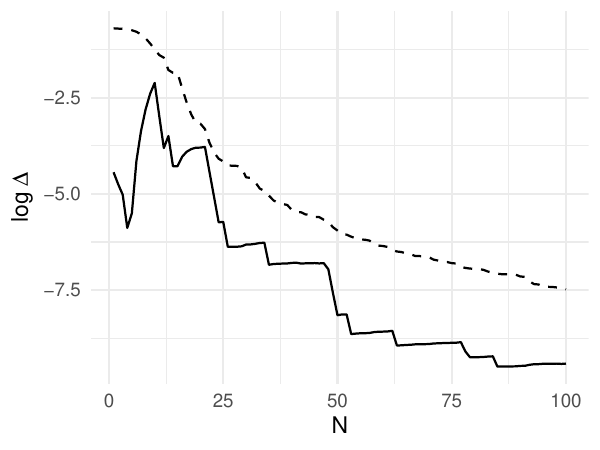}
\caption{$d = 2$, $\kappa = 3$.}
\label{fig:vmf-orthant-d1-kappa3}
\end{subfigure}

\begin{subfigure}{0.32\textwidth}
\includegraphics[width=\textwidth]{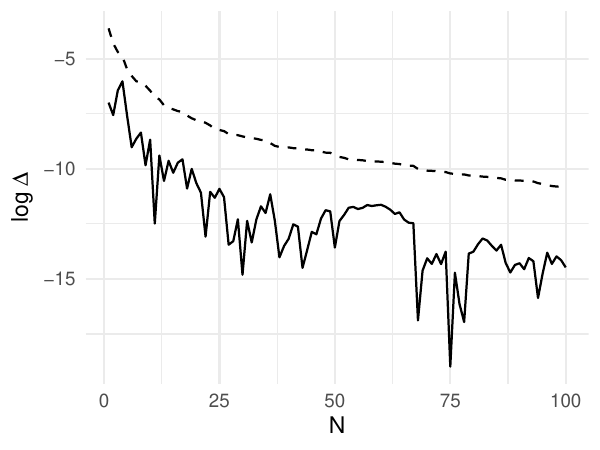}
\caption{$d = 4$, $\kappa = 0.3$.}
\label{fig:vmf-orthant-d2-kappa1}
\end{subfigure}
\begin{subfigure}{0.32\textwidth}
\includegraphics[width=\textwidth]{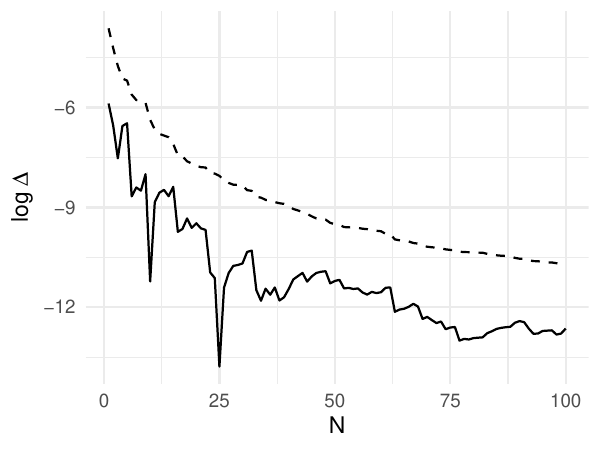}
\caption{$d = 4$, $\kappa = 1$.}
\label{fig:vmf-orthant-d2-kappa2}
\end{subfigure}
\begin{subfigure}{0.32\textwidth}
\includegraphics[width=\textwidth]{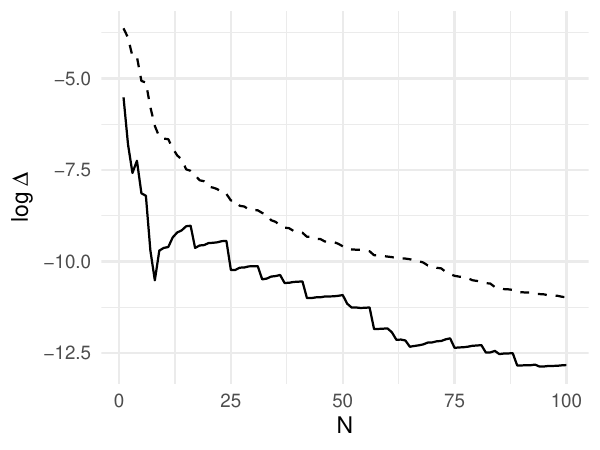}
\caption{$d = 4$, $\kappa = 3$.}
\label{fig:vmf-orthant-d2-kappa3}
\end{subfigure}

\begin{subfigure}{0.32\textwidth}
\includegraphics[width=\textwidth]{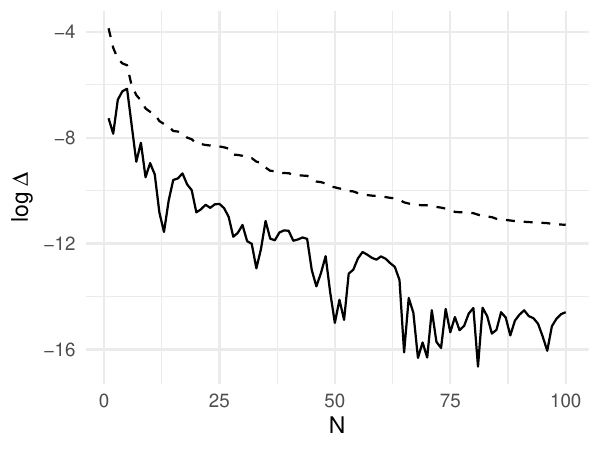}
\caption{$d = 5$, $\kappa = 0.3$.}
\label{fig:vmf-orthant-d3-kappa1}
\end{subfigure}
\begin{subfigure}{0.32\textwidth}
\includegraphics[width=\textwidth]{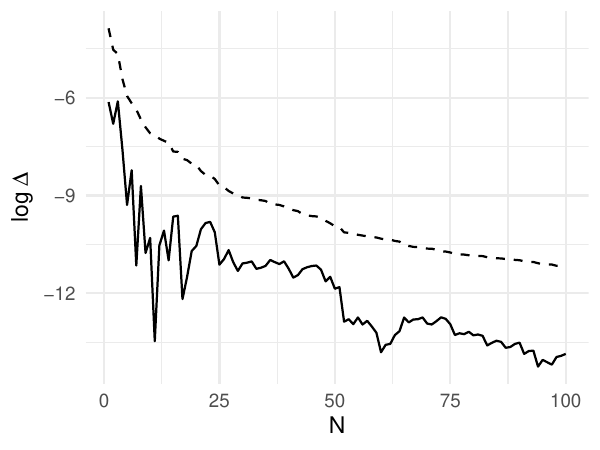}
\caption{$d = 5$, $\kappa = 1$.}
\label{fig:vmf-orthant-d3-kappa2}
\end{subfigure}
\begin{subfigure}{0.32\textwidth}
\includegraphics[width=\textwidth]{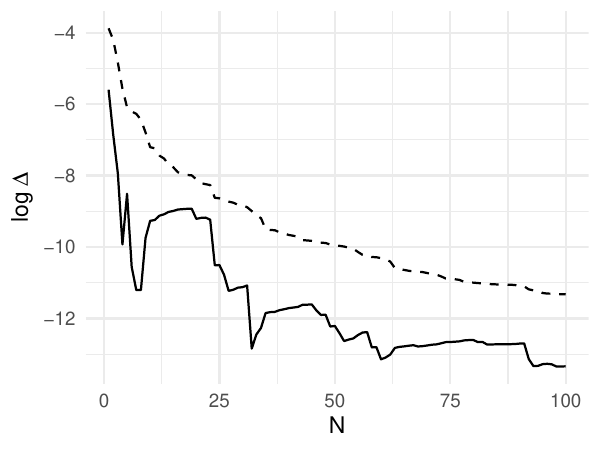}
\caption{$d = 5$, $\kappa = 3$.}
\label{fig:vmf-orthant-d3-kappa3}
\end{subfigure}
\caption{Log of approximation error $\Delta$ using linear VWS proposal (solid curve) versus upper bound (dashed curve) given in \eqref{eqn:vmf-orthant-bound}.}
\label{fig:vmf-orthant}
\end{figure}

\subsection{Knot Selection Study}
\label{sec:vmf-knot-selection-study}

The selection of knots $\alpha_1, \ldots, \alpha_{N-1}$ to partition a univariate domain $\Omega$ into regions $\mathscr{D}_j = (\alpha_{j-1}, \alpha_j]$ with $j = 1, \ldots, N$ for use in rejection sampling is an important consideration which may have a large impact on the rejection rate of the sampler. Algorithm~\ref{alg:refine} has been presented as a primary approach to knot selection. This method directly aims to minimize bound \eqref{eqn:bound} by sequentially bifurcating regions with probabilities proportional to their contributions $\rho_1, \ldots, \rho_N$. A ``greedy'' variation of Algorithm~\ref{alg:refine} can be obtained by instead selecting region $\ell = \argmax \{\rho_1, \ldots, \rho_{N_0 + j} \}$ on line~\ref{line:selection}. For Algorithm~\ref{alg:refine} and its greedy variant, the following remark describes how we have elected to bifurcate intervals which are not bounded.

\begin{remark}
In line~\ref{line:midpoint}, we assume the midpoint of $\mathscr{D}_\ell = (\alpha_{\ell-1}, \alpha_{\ell}]$ to be
\begin{align*}
\alpha^* =
\begin{cases}
0, & \text{if $\alpha_{\ell-1} = -\infty$ and $\alpha_\ell = \infty$}, \\
\alpha_\ell - |\alpha_\ell| - 1, & \text{if $\alpha_{\ell-1} = -\infty$ and $\alpha_\ell < \infty$}, \\
\alpha_{\ell-1} + |\alpha_{\ell-1}| + 1, & \text{if $\alpha_{\ell-1} > -\infty$ and $\alpha_\ell = \infty$}, \\
(\alpha_{\ell-1} + \alpha_\ell) / 2, & \text{otherwise},
\end{cases}
\end{align*}
when the target is a continuous distribution; i.e., the arithmetic midpoint when both endpoints are finite and zero when both are infinite. When one of the two is finite, $\alpha^*$ is taken to be a shifted version of that endpoint. A similar bifurcation is used when the target is a discrete distribution, but with $\alpha^* = \lceil (\alpha_{\ell-1} + \alpha_\ell) / 2 \rceil$ in the case that both endpoints are finite. If no values from the support are within $(\alpha_{\ell-1}, \alpha_\ell]$, the region should be excluded from further bifurcation; this follows from $\rho_\ell = 0$.
\end{remark}

Two other commonly considered approaches include using equally spaced knots and selecting knots to produce regions with equal probability \citep{Hormann2002}. Taking equally spaced knots is one of the simplest methods to implement, and may be carried out by selecting knots such that $|\mathscr{D}_j| = |\Omega| / N$ for $j = 1, \ldots, N$. Selecting knots with equal probability may be accomplished for a given $N$ by partitioning support $\Omega$ into $N$ regions with $\overline{\xi}_j \approx \psi / N$. Note that such knots are ``equal probability'' in the sense that regions $\mathscr{D}_1, \ldots, \mathscr{D}_N$ will be drawn with approximately equal probability during rejection sampling; however, their contributions $\rho_1, \ldots, \rho_N$ to bound \eqref{eqn:bound} will vary when $\underline{\xi}_1, \ldots, \underline{\xi}_N$ vary. Algorithm~\ref{alg:refine-equal-probs} presents a method to sequentially partition $\Omega$ in this way. Note that $\sum_{j=1}^N \overline{\xi}_j \geq \psi$ for any choice of knots; therefore, the $N$th region containing the ``remainder'' is likely to have an $\overline{\xi}_N$ which is somewhat larger than $\psi / N$. The minimization on line~\ref{line:equal-probs-opt} of Algorithm~\ref{alg:refine-equal-probs} is carried out numerically using Brent's algorithm \citep{Brent1973} in the present section, as the support $\Omega$ is bounded.

\begin{algorithm}
\caption{Rule of thumb for equal probability knot selection.}
\label{alg:refine-equal-probs}
\hspace*{\algorithmicindent}\textbf{Input}: number of regions $N$.
\begin{algorithmic}[1]
\State Let $\mathscr{D} \leftarrow \Omega$.
\For{$j = 1, \ldots, N-1$}
\State Let $\overline{\xi}(x) = \int_{\mathscr{D}} \ind(s \leq x) \overline{w}(s) g(s) d\nu(s)$.
\State Let $\alpha_j = \argmin_{x \in \mathscr{D}} [ \overline{\xi}(x) - \psi/N ]^2$ and $\overline{\xi}_j = \overline{\xi}(\alpha_j)$. \label{line:equal-probs-opt}
\State Let $\mathscr{D} \leftarrow \Omega \cap (\alpha_j, \infty]$.
\EndFor
\State \Return $(\alpha_1, \ldots, \alpha_{N-1})$.
\end{algorithmic}
\end{algorithm}

We present a small study to compare the performance of VWS with target density \eqref{eqn:vmf-target} using several knot selection methods. Nine combinations of $(\kappa, d)$ are used with $\kappa \in \{ 0.1, 10 \}$ and $d \in \{ 2, 3, 4 \}$ to mimic the setup in Section~\ref{sec:vmf-generation}. The study considers four knot selection methods including: probabilistic sequential selection via Algorithm~\ref{alg:refine} and its greedy variant, equally spaced knots, and equal probability knots via Algorithm~\ref{alg:refine-equal-probs}.

Figures~\ref{fig:vmf-knots-const} and \ref{fig:vmf-knots-linear} compare the four knot selection methods under constant and linear VWS, respectively. Both sets of plots evaluate the rejection rate from a single region up to 100 regions. Sequential partitioning was repeated 100 times with the probabilistic selection method; the pointwise median is plotted along with an interval (which appears as a thin blue band) based on the 2.5\% and 97.5\% quantiles. Each plot also includes rejection rates from the UW rejection sampler described in Section~\ref{sec:vmf-generation} as a baseline comparison, which have been computed empirically using 50,000 draws. 

In general, the probabilistic and greedy sequential selection approaches are more efficient than either the equally spaced or equal probability methods in both the constant and linear VWS cases. The equally spaced method has much lower efficiency than the equal probability method for $d=2$; however, for smaller $\kappa$ and larger $d$, equally spaced selection outperforms equal probability.

\begin{figure}
\centering
\begin{subfigure}{0.32\textwidth}
\includegraphics[width=\textwidth]{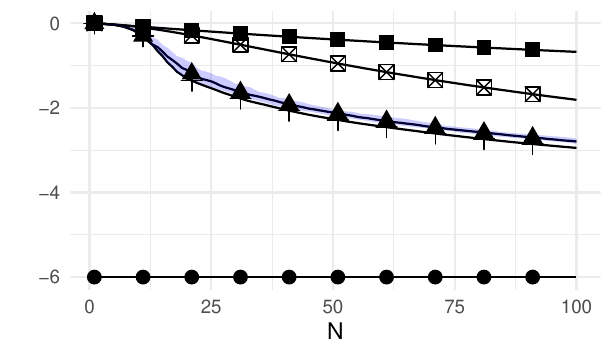}
\caption{$d = 2$, $\kappa = 0.1$.}
\label{fig:vmf-knots-constant-d1-kappa1}
\end{subfigure}
\begin{subfigure}{0.32\textwidth}
\includegraphics[width=\textwidth]{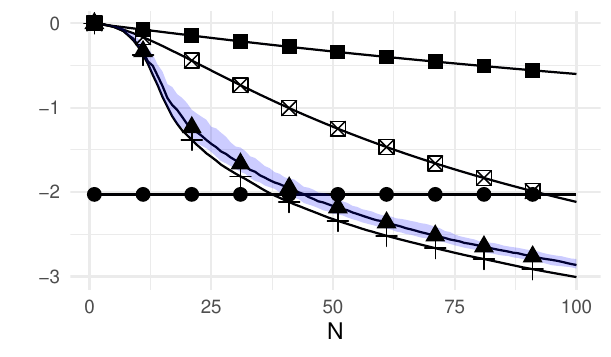}
\caption{$d = 2$, $\kappa = 1$.}
\label{fig:vmf-knots-constant-d1-kappa2}
\end{subfigure}
\begin{subfigure}{0.32\textwidth}
\includegraphics[width=\textwidth]{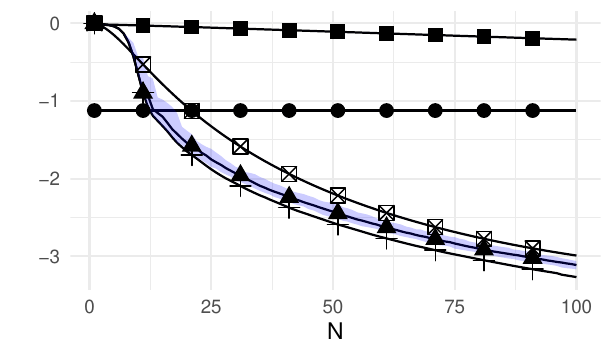}
\caption{$d = 2$, $\kappa = 10$.}
\label{fig:vmf-knots-constant-d1-kappa3}
\end{subfigure}

\begin{subfigure}{0.32\textwidth}
\includegraphics[width=\textwidth]{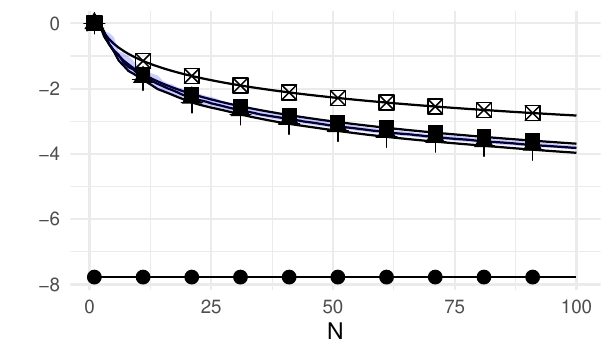}
\caption{$d = 4$, $\kappa = 0.1$.}
\label{fig:vmf-knots-constant-d2-kappa1}
\end{subfigure}
\begin{subfigure}{0.32\textwidth}
\includegraphics[width=\textwidth]{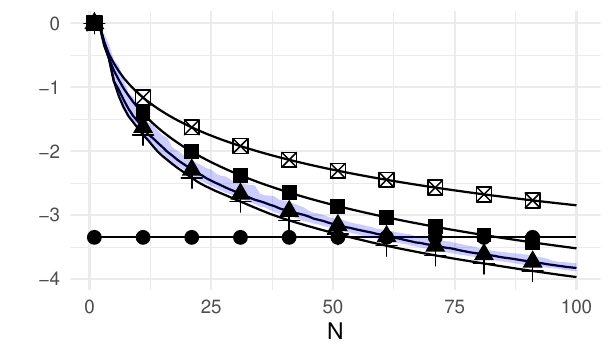}
\caption{$d = 4$, $\kappa = 1$.}
\label{fig:vmf-knots-constant-d2-kappa2}
\end{subfigure}
\begin{subfigure}{0.32\textwidth}
\includegraphics[width=\textwidth]{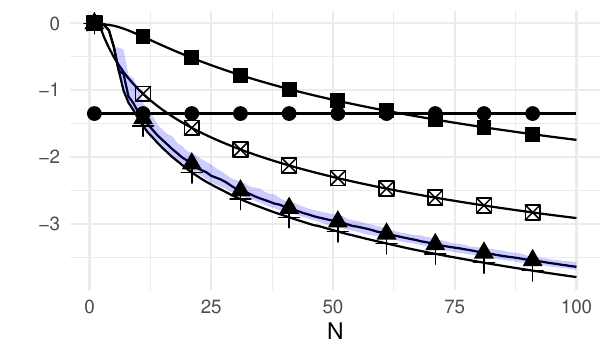}
\caption{$d = 4$, $\kappa = 10$.}
\label{fig:vmf-knots-constant-d2-kappa3}
\end{subfigure}
	
\begin{subfigure}{0.32\textwidth}
\includegraphics[width=\textwidth]{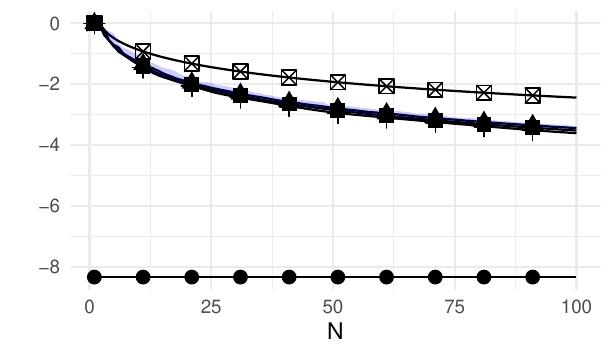}
\caption{$d = 5$, $\kappa = 0.1$.}
\label{fig:vmf-knots-constant-d3-kappa1}
\end{subfigure}
\begin{subfigure}{0.32\textwidth}
\includegraphics[width=\textwidth]{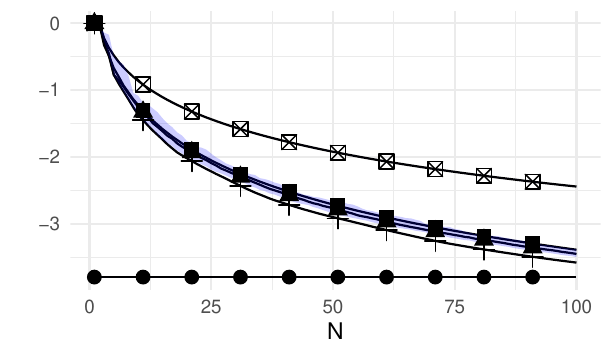}
\caption{$d = 5$, $\kappa = 1$.}
\label{fig:vmf-knots-constant-d3-kappa2}
\end{subfigure}
\begin{subfigure}{0.32\textwidth}
\includegraphics[width=\textwidth]{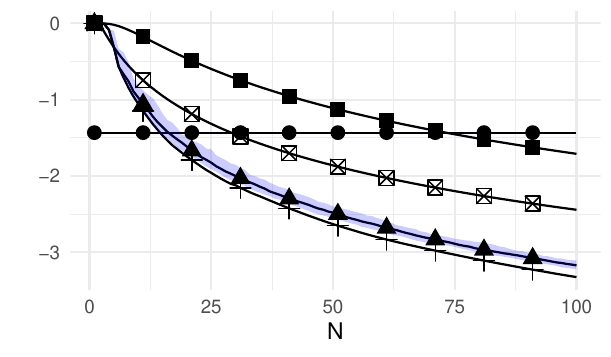}
\caption{$d = 5$, $\kappa = 10$.}
\label{fig:vmf-knots-constant-d3-kappa3}
\end{subfigure}
\caption{Log of rejection probability $\log(1 - \psi / \psi_N)$ using constant VWS under four knot selection methods: equally spaced ($\blacksquare$), equal probability ($\boxtimes$), probabilistic sequential ($\blacktriangle$), and greedy sequential ($+$). UW ($\CIRCLE$) is shown for reference.}
\label{fig:vmf-knots-const}
\end{figure}

\begin{figure}
\centering
\begin{subfigure}{0.32\textwidth}
\includegraphics[width=\textwidth]{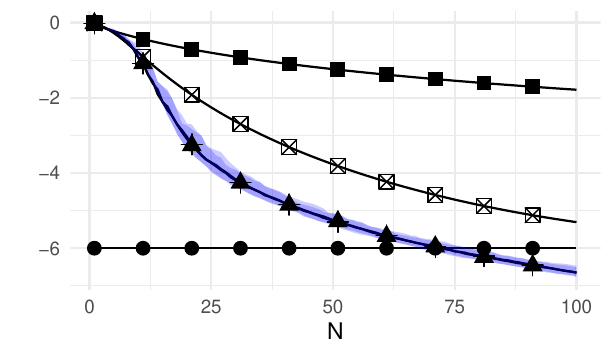}
\caption{$d = 2$, $\kappa = 0.1$.}
\label{fig:vmf-knots-linear-d1-kappa1}
\end{subfigure}
\begin{subfigure}{0.32\textwidth}
\includegraphics[width=\textwidth]{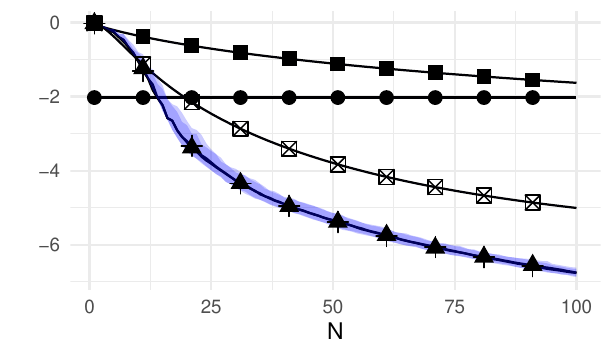}
\caption{$d = 2$, $\kappa = 1$.}
\label{fig:vmf-knots-linear-d1-kappa2}
\end{subfigure}
\begin{subfigure}{0.32\textwidth}
\includegraphics[width=\textwidth]{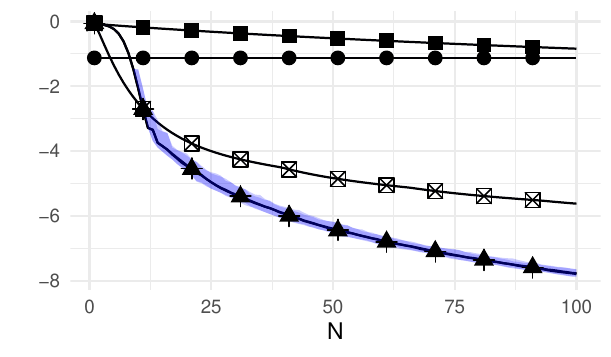}
\caption{$d = 2$, $\kappa = 10$.}
\label{fig:vmf-knots-linear-d1-kappa3}
\end{subfigure}

\begin{subfigure}{0.32\textwidth}
\includegraphics[width=\textwidth]{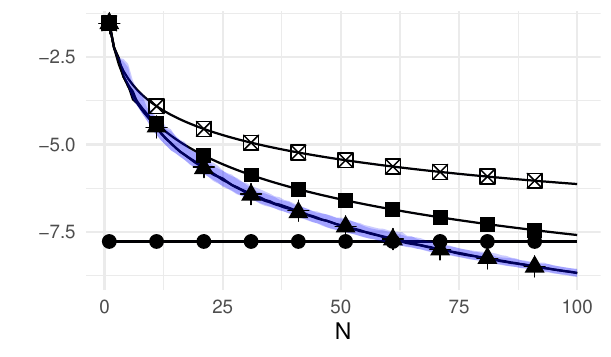}
\caption{$d = 4$, $\kappa = 0.1$.}
\label{fig:vmf-knots-linear-d2-kappa1}
\end{subfigure}
\begin{subfigure}{0.32\textwidth}
\includegraphics[width=\textwidth]{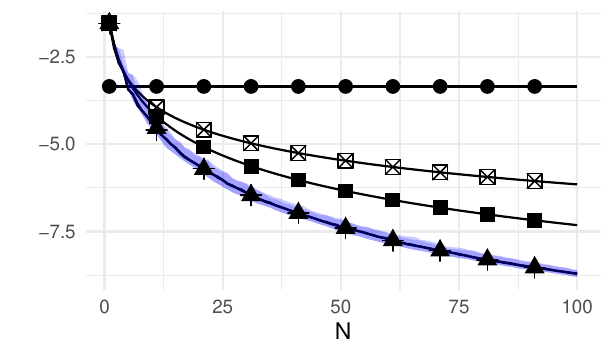}
\caption{$d = 4$, $\kappa = 1$.}
\label{fig:vmf-knots-linear-d2-kappa2}
\end{subfigure}
\begin{subfigure}{0.32\textwidth}
\includegraphics[width=\textwidth]{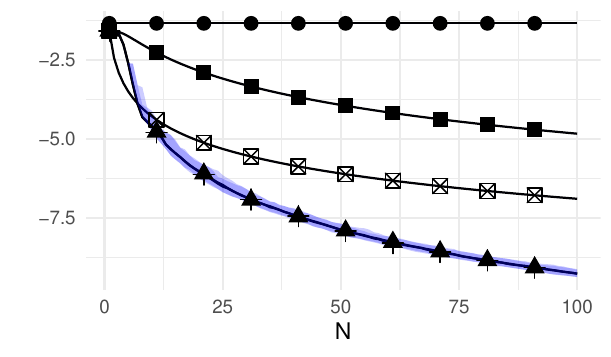}
\caption{$d = 4$, $\kappa = 10$.}
\label{fig:vmf-knots-linear-d2-kappa3}
\end{subfigure}

\begin{subfigure}{0.32\textwidth}
\includegraphics[width=\textwidth]{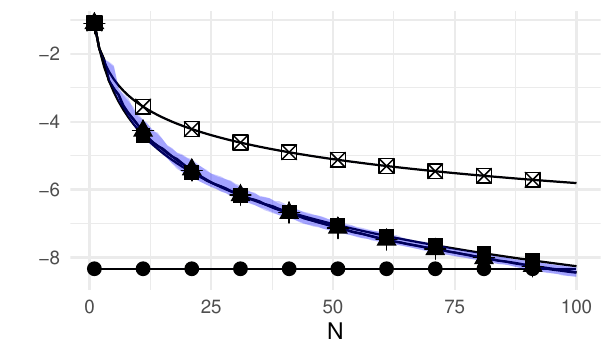}
\caption{$d = 5$, $\kappa = 0.1$.}
\label{fig:vmf-knots-linear-d3-kappa1}
\end{subfigure}
\begin{subfigure}{0.32\textwidth}
\includegraphics[width=\textwidth]{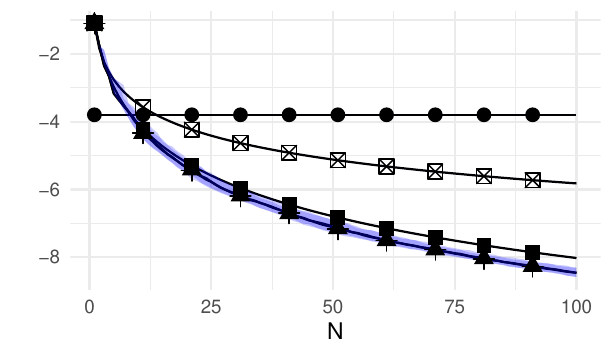}
\caption{$d = 5$, $\kappa = 1$.}
\label{fig:vmf-knots-linear-d3-kappa2}
\end{subfigure}
\begin{subfigure}{0.32\textwidth}
\includegraphics[width=\textwidth]{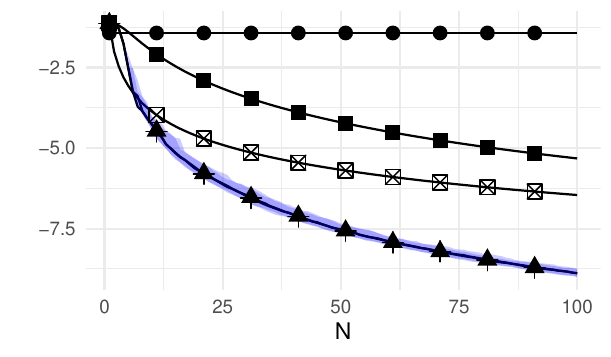}
\caption{$d = 5$, $\kappa = 10$.}
\label{fig:vmf-knots-linear-d3-kappa3}
\end{subfigure}
\caption{Log of rejection probability $\log(1 - \psi / \psi_N)$ using linear VWS under four knot selection methods: equally spaced ($\blacksquare$), equal probability ($\boxtimes$),  probabilistic sequential ($\blacktriangle$), and greedy sequential ($+$). UW ($\CIRCLE$) is shown for reference.}
\label{fig:vmf-knots-linear}
\end{figure}

\end{document}